\documentclass[11pt,a4paper]{article}
\usepackage{fullpage,color,amsthm,amssymb,graphicx,amsmath,xspace,wrapfig,url}
\usepackage{hyperref}\hypersetup{colorlinks=false, pdfborder={0 0}}\usepackage[numbers,sort&compress]{natbib}\usepackage{microtype}\usepackage[small]{caption}
\newtheorem{theorem}{Theorem}[section]\newtheorem{lemma}[theorem]{Lemma}\newtheorem{claim}[theorem]{Claim}
\theoremstyle{definition}\newtheorem{definition}[theorem]{Definition}
\renewcommand{\-}{\textrm{-}}\newcommand{\e}{\emph}\renewcommand{\cal}[1]{\ensuremath{\mathcal{{#1}}}\xspace}\newcommand\eps{\ensuremath{\varepsilon}\xspace}
\newcommand{\ddp}{discrete Dubins path\xspace}\newcommand{\ddps}{discrete Dubins paths\xspace}
\newcommand\cw{clockwise\xspace}\newcommand\ccw{counterclockwise\xspace}
\renewcommand{\th}{\ensuremath{\theta}\xspace}\renewcommand{\ll}{\ensuremath{\ell}\xspace}
\renewcommand{\P}{\ensuremath{P}\xspace}
\renewcommand{\|}[1]{\ensuremath{|#1|}\xspace}
\newcommand{\conf}[1]{\ensuremath{\mathcal{#1}=(\lowercase{#1},#1)}\xspace}
\newcommand{\uu}{\ensuremath{\mathcal{U}}\xspace}\newcommand{\vv}{\ensuremath{\mathcal{V}}\xspace}
\renewcommand{\u}{\ensuremath{u}\xspace}\renewcommand{\v}{\ensuremath{v}\xspace}
\newcommand{\U}{\ensuremath{U}\xspace}\newcommand{\V}{\ensuremath{V}\xspace}
\newcommand{\ccone}{\ensuremath{\mathcal{C}_1}\xspace}\newcommand{\cctwo}{\ensuremath{\mathcal{C}_2}\xspace}

\newcommand\g{\ensuremath{\gamma}\xspace}
\renewcommand\a[2]{\ensuremath{\angle{{#1}{#2}}}\xspace}

\newcommand\pb{\ensuremath{\varphi_b}\xspace}\newcommand\pa{\ensuremath{\varphi_a}\xspace}

  \renewenvironment{thebibliography}[1]{%
    \begin{oldthebibliography}{#1}%
      \setlength{\itemsep}{0ex}%
  }%
  {%
    \end{oldthebibliography}%
  }

\begin{document}\title{Discrete Dubins Paths}\author{Sylvester Eriksson-Bique\thanks{Courant Institute \texttt{syerikss@uw.edu}} \and David Kirkpatrick\thanks{University of British Columbia \texttt{kirk@cs.ubc.ca}} \and Valentin Polishchuk\thanks{University of Helsinki \texttt{polishch@cs.helsinki.fi}}}\date{}\maketitle
\begin{abstract}A \e{Dubins path} is a shortest path with bounded curvature. The seminal result in non-holonomic motion planning is that (in the absence of obstacles) a Dubins path consists either from a circular arc followed by a segment followed by another arc, or from three circular arcs [Dubins, 1957]. 
Dubins original proof uses advanced calculus; later, Dubins result was reproved using control theory techniques [Reeds and Shepp, 1990], [Sussmann and Tang, 1991], [Boissonnat, C\'er\'ezo, and Leblond, 1994].

We introduce and study a discrete analogue of curvature-constrained motion. We show that shortest ``bounded-curvature'' polygonal paths have the same structure as Dubins paths. The properties of Dubins paths follow from our results as a limiting case---this gives a new, ``discrete'' proof of Dubins result. 
\end{abstract}
\section{Introduction}
Curvature-constrained paths are a fundamental tool in planning motion with bounded turning radius. Because the paths are \e{smooth}, they may look more appealing and realistic than \e{polygonal} paths. Still, polygonal paths are a much more common model in geometry, exactly because of their discrete nature: polygonal paths are more natural to plan, describe and follow. For instance, in air traffic management---one of our motivating applications---an aircraft flight plan (``waypoints list'') is represented on the strategic level by a polygonal path whose vertices are the GPS waypoints; see Fig.~\ref{fig:flighplans} in Appendix~\ref{app:flightpaths} for examples. The actual smoothly turning trajectory at a waypoint is decided by the pilot on the tactical level when executing the turn (see~\cite{krozel} for more on curvature-constrained route planning in air transport). We are thus motivated to work out a \e{discretized} model of curvature-constrained motion.

\begin{wrapfigure}{r}{.4\columnwidth}\centering\includegraphics{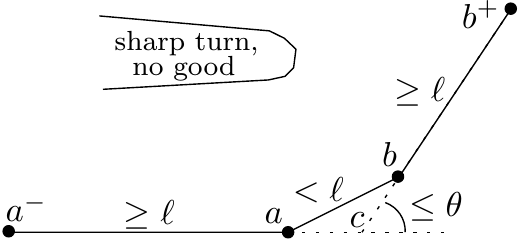}\end{wrapfigure}
\paragraph{Angles and Lengths}The first and foremost constraint put on a ``bounded-curvature'' polygonal path \P is that it never turns ``too sharply''. Formally, let $0\le\th\le\pi/2$ be a given angle (assume that \th perfectly divides $2\pi$); we require that at its every vertex \P turns by at most~\th. Such a restriction alone, however, does not yet guarantee that \P will serve as a bona fide discrete analogue of a bounded-curvature path: many short successive edges of \P, each turning only slightly, can simulate a sharp turn. Thus we also introduce a length constraint.

Specifically, let $\ll>0$ be a given length; call a path edge \e{short} if it is shorter than~\ll. We require that in a discrete curvature-constrained path no two short edges are adjacent, and that the turn from one edge adjacent to a short edge to the other edge adjacent to the short edge is at most~\th. More specifically, let $ab$ be a short non-terminal edge; let $a^-a,bb^+$ be the (non-short) edges adjacent to $ab$. Suppose that $ab$ is a \e{non-inflection} edge, i.e., that $a^-a$ and $bb^+$ lie on the same side of the supporting line of $ab$. Let $c$ be the point of intersection of supporting lines of $a^-a$ and $bb^+$. We require that the supplementary angle of the angle $acb$ is at most~\th.

Summarizing, we arrive at the following definition:
\begin{definition}\label{def:dccp}A \e{discrete curvature-constrained path} is a polygonal path with the following properties:
\begin{list}{}{}
\item[\e{Turn constraints:}]\label{constr:turn}the turn at every vertex is at most \th;
\item[\e{Length constraints:}]\label{constr:length}there are no adjacent short edges;
\item[\e{Turn-over-length constraints:}]\label{constr:turn-length}for any short non-inflection edge $e$ the turn from one non-short edge adjacent to $e$ to the other non-short edge adjacent to $e$ is at most~\th.
\end{list}\end{definition}
\begin{figure}[b]\centering\includegraphics{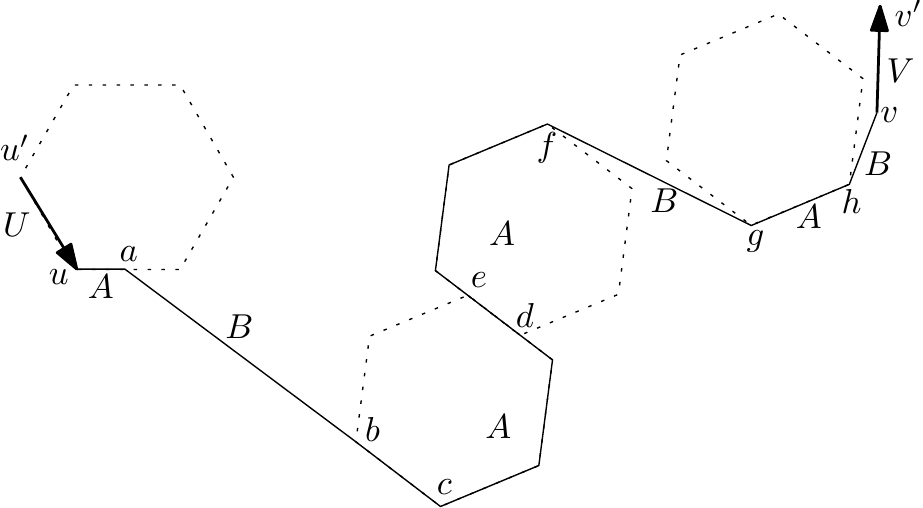}\caption{A \uu-\vv path of type $ABAABAB$ (see Section~\ref{sec:typeANDcanon} for definition of the type and classification of different parts of the path.)}\label{fig:example}\end{figure}
\paragraph{Configurations}A \e{configuration} \conf{U} is a pair (point, direction), where the direction is
a length-\ll vector. If a robot follows a curvature-constrained path from \conf{U}, then the robot starts at \u having velocity~\U. In our discrete setting we require that the robot starts at $\u$ and ``does not turn too sharply'' there, which we formalize as follows: imagine that a ``pre-edge'' $\u'\u$ is attached to \P at \u, where $\u'=\u-\U$ (Fig.~\ref{fig:example}); we require that this longer path is also a discrete curvature-constrained path, i.e., that the turn at $u$ is at most \th, and if the first edge of \P is short, the turn from the pre-edge to the second edge of \P is at most~\th. (The pre-edge is not an edge of \P; we attach it only 
conceptually, to formalize the restriction of \P's turn at \u.) We say that \P \e{starts at} \uu. Similarly, let \conf{V} be another configuration. We say that \P \e{ends at} \vv if $\v$ is the last vertex of \P and the path obtained by appending to \P the segment $\v\v'$, where $\v'=\v+\V$, is still a curvature-constrained path.
\begin{definition}A \e{\ddp,} or \e{discrete curvature-constrained
geodesic ($DC^2G$)}, is a shortest discrete curvature-constrained path that starts at \uu and ends at~\vv.\end{definition}
\subsection{Related work}The study of curvature-constrained path planning has a rich history that long predates and goes well beyond robot motion planning, for example the work of Markov~\cite{m-1887} on the construction of railway segments. A fundamental result in curvature-constrained motion planning, due to Dubins~\cite{d-cmlca-57} (hence the name---\e{Dubins paths}---for shortest bounded-curvature paths), states that, in the absence of obstacles, shortest paths
in the plane,
whose mean curvature is bounded by 1, are one of two types: a (unit) circular arc followed by a line segment followed by another arc, or a sequence of three circular arcs.
This result has been rederived and extended using techniques from optimal control theory in \cite{rs-opctg-90,bcl-spbcp-94,st-sprsc-91}. Variations and generalizations of the problem were studied in \cite{bsbl-spsdn-94,bb-arcto-94,f-dpmrr-04,s-s3dpp-95,rw-nudkm-96,crashing,airplane,tsp,mbp-vmcilfdlv-08,surfaces,receding,chitsazPhD,dsdw-nmrlot-99,lumelsky,dir,dolinskaya}.



The books \cite{l-rmp-91,lc-nmp-92} are general references; for some very recent work on bounded-curvature paths see \cite{needle,fb-gdcmtsbfrtp-10,vendittelli,minsteer,needleinverse,et-itapocuavuga-11,reachability,lengthVScurv}. A discretization of curvature-constrained motion was studied by Wilfong \cite{w-mpav-88,w-spav-89}; his setting is different from ours since he considered discretized \e{environment}, and not discrete paths. A practical way of producing paths with length and turn constraints is presented in~\cite{szczerba}.
\subsection{Our approach}We study properties of \ddps, as well as their continuous counterparts. In Section~\ref{sec:edges} we examine different types of edges of a \ddp, and conclude that \e{at least} half of the edges of the path must have length \e{exactly}~\ll. We go on to proving that the path has \e{at most one} edge with length \e{strictly greater} than \ll, and \e{at most two} inflection edges; moreover, if the path does contain a longer-than-\ll edge, then it does not have inflection edges. Overall we obtain that a \ddp ``mostly'' consists of non-inflection length-\ll edges, which prompts us to define ``discrete circular arcs'' as chains of length-\ll edges with turns \th from one edge to the next (Section~\ref{sec:arcs}). We use the arcs to define the ``type'' of a \ddp similarly to how Dubins defined the type of a smooth path (Section~\ref{sec:typesDDP}): we split the path into (discrete) circular arcs and ``bridges''. Using the freedom of bridges, formalized in Section~\ref{sec:typesDDP}, together with the results from Section~\ref{sec:edges}, we eliminate certain types as infeasible for a \ddp, and arrive at the characterization of \ddps in terms of the type (Theorem~\ref{dccg}).

In Section~\ref{sec:limit} we obtain a new proof of Dubins' result for the type of smooth curvature-constrained shortest paths as the limiting case of our results for \ddps (Theorem~\ref{dubins}). In comparison to the earlier proofs of Dubins' characterization (which used calculus, Pontryagin's principle, Arzel\'a--Ascoli theorem, etc.) \cite{d-cmlca-57,rs-opctg-90,bcl-spbcp-94,st-sprsc-91}, our argument is more ``discrete'' and thus is
arguably simpler than the previous ones.
Similar methods to prove ``continuous'' results using discretization were already used by Shur in his paper of 1921; interestingly, exactly these problems, considered by Schur~\cite{asc-usek-21} and Schmidt~\cite{esc-uebr-25},
led Dubins to his result.



\section{Edges of \ddps}\label{sec:edges}We now start our investigation of \ddps edges. Recall that an internal edge $e$ of a polygonal path is an \e{inflection} edge if the edges adjacent to $e$ lie on the opposite sides of (the supporting line of)~$e$. Inflection edges receive special treatment from us (which shows up already in the definition of discrete curvature-constrained paths -- we ignored turn-over-length constraints for short inflection edges) because, intuitively, a path does not ``gain'' curvature while traversing an inflection edge. This is due to an equivalent definition of an inflection edge: $e$ is an inflection edge if the path makes opposite turns at the vertices of~$e$.

Recall also that an edge is \e{short} if it is shorter than \ll. We say that an edge is \e{normal} (resp.\ \e{long}) if its length is exactly (resp.\ larger than)~\ll.

In this section we present a series of results showing
that, with only a small constant number of exceptions, a \ddp consists entirely of non-inflection normal edges.

\subsection{Normal edges are ``frequent''}Our first two observations are that in a \ddp long, short and inflection edges are separated by normal edges. The proofs of the next two lemmas are elementary applications of shortcutting arguments; we nevertheless spell out the full proofs since they serve as examples of the kind of arguments we use throughout.
\begin{lemma}\label{lem:ln}A long or a short edge can be adjacent only to a normal edge.\end{lemma}
\begin{wrapfigure}{r}{.2\columnwidth}\centering\includegraphics{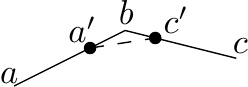}\end{wrapfigure}
\noindent\e{Proof. }
Suppose that $ab,bc$ are two adjacent long edges. Choose points $a'\in ab,b'\in c'c$ close to $b$ so that $aa'$ and $c'c$ are long, and $a'c'$ is short. Replacing the subpath $abc$ with $aa'c'c$ results in a shorter path, which is a feasible discrete curvature-constrained path. Indeed, the turns at $a'$ and at $c'$ are smaller than the original turn at $b'$ -- hence the turn constraints are satisfied (the other turns of the path do not change because of the shortcut). The length constraint is satisfied because $aa'$ and $c'c$ are both long. Finally, the turn-over-length constraint is satisfied for $a'c'$ because the turn from $aa'$ to $c'c$ is the same as it was from $ab$ to $bc$.

\begin{wrapfigure}{r}{.2\columnwidth}\centering\includegraphics{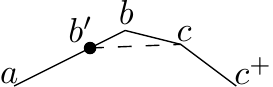}\end{wrapfigure}
Suppose now that a long edge $ab$ is adjacent to a short edge $bc$. Select a point $b'\in ab$ close to $b$ so that $ab'$ is still long and $b'c$ is still short. The path going through $b'$ instead of $b$ is shorter; let us check that it is still feasible: The turn at $b'$ is smaller than the turn at $b$ was, so the turn constraints are satisfied. The length constraints are obviously satisfied. Finally, the turn from $ab'$ to the edge $cc^+$ following $b'c$ is the same as the turn from $ab$ to $cc^+$, so the turn-over-length constraints are satisfied too.
\qed
\begin{lemma}\label{lem:infl}An inflection edge can be adjacent only to a non-inflection normal edge.\end{lemma}
\begin{proof}Let $bc,cd$ be the edges that follow an inflection edge $ab$. If $bc$ is not normal, rotate $ab$ around $a$ while sliding $b$ along $bc$ (Fig.~\ref{fig:infl}). Clearly, this shortens the path. To see that the path remains feasible, let us check that all constraints in the definition of a discrete curvature-constrained path are satisfied:
\begin{list}{}{}
\item[\e{Turn constraints}]The inflection edge endpoints are the only vertices the turns at which change due to the modification of the path; however, both turns only decrease.
\item[\e{Length constraints}]The only edge whose length decreases due to the rotation is $bc$. If it was short before the rotation, then it was adjacent to non-short edges, i.e., both $ab$ and $cd$ were non-short. The length of $ab$ only increases with the modification, and the length of $cd$ stays the same; thus $bc$ is still adjacent to non-short edges. On the other hand, if $bc$ was long, we can move $b$ by small enough distance to ensure that $bc$ stays long.
\item[\e{Turn-over-length constraints}]Again, the only edge to worry about is $bc$. If it is an inflection edge, the constraint is irrelevant for it. If $bc$ is non-inflection, the turn from $ab$ to $cd$ only decreases with the rotation.
\end{list}
Thus, $bc$ must be normal. Now, if $bc$ is a terminal edge ($c=v$), then it is not an inflection edge, and we are done. Otherwise $bc$ is an internal edge. However, if $bc$ is an inflection edge, then the vertex $b$ could be bypassed by connecting $a$ to $c$ directly, shortening the path.
\begin{figure}\centering\includegraphics{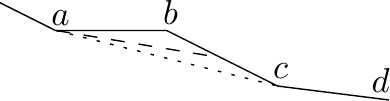}\caption{If $bc$ is long (resp.\ short), then $b$ can be slid towards $c$ slightly while keeping $bc$ long (resp.\ short) -- rotating $ab$ around $a$ decreases the turns at $a$ and $b$, and shortens the path (the dashed segment is the shortcut). If $bc$ is normal, then $ac$ (dotted) is long; since also $\angle acd>\angle bcd$, we have that $ac$ is a feasible shortcut.}\label{fig:infl}\end{figure}\end{proof}

\subsection{``Freedom'' of inflection and long edges}\label{sec:freedom}The proof of Lemma~\ref{lem:infl} used the possibility to move $b$ locally along $bc$ while rotating the edge $e=ab$ about $a$. In the next lemma, we show that $b$ has actually much more freedom than just moving along the adjacent edge: there is a whole halfplane of ``admissible'' directions.
\begin{definition}\label{def:admissible}A vector $T$ is \e{admissible} for $e$ if, when applied to $b$, the vector points to the side of (the supporting line of) $e$ on which $bc$ lies (Fig.~\ref{fig:Ifreedom}, left).\end{definition}
Let $P'$ be the part of \P starting from $b$, and let $\P(T)$ be the path obtained from \P by rigidly translating $P'$ by $T$, and rotating $e$ around $a$ so as to keep the path connected (i.e., by connecting $a$ to $b(T)=b+T$, the new position $b$); Fig.~\ref{fig:Ifreedom}, right.
\begin{figure}\centering\includegraphics{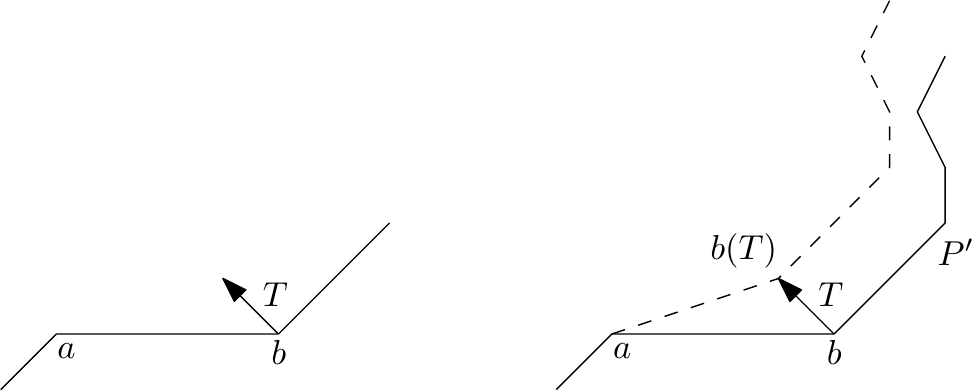}\caption{Left: If $e$ is horizontal and $b$ is a left-turn vertex, then admissible vectors point up. Right: The modification.}\label{fig:Ifreedom}\end{figure}
\begin{lemma}\label{lem:Ifreedom}For any admissible vector $T$ there exists $\eps>0$ such that $\P(\eps T)$ is a feasible discrete curvature-constrained path.\end{lemma}
\begin{proof}The inflection edge endpoints are the only vertices the turns at which may change due to the modification of the path. It is easy to see that if $T$ is admissible, the turns at both endpoints of the edge decrease. Thus, the turn and the turn-over-length constraints are satisfied. As for the length constraints, the inflection edge is the only edge whose length may decrease due to the modification; for that to possibly break the length constraint, the edge must be adjacent to a short edge (and also must itself become short in $P(\eps T)$ while being normal in \P{} -- otherwise, if the edge was long in \P, it will remain long also in $P(\eps T)$, for a sufficiently small \eps). But having an inflection edge adjacent to a short edge contradicts Lemma~\ref{lem:infl}.\end{proof}
Note that even though \P was assumed to be a \e{shortest} path (optimality of \P was used in proving that the length and the turn-over-length constraints hold for neighbors of the bridge), we do not claim that $\P(\eps T)$ will be the shortest path between its start and destination; we only claim that $P(\eps T)$ is \e{feasible}.

We now argue that long non-inflection edges have even more freedom than inflection edges: no matter where a vertex of such an edge is moved, the edge can stay connected to the vertex while maintaining path feasibility. There is a price to pay for the connectivity though: the long edge does not stay intact, but gets ``broken'' into a normal edge and the remainder. That is, an additional vertex of a path appears at a point at distance \ll from one of the edge endpoints (such addition of a vertex in the middle of a long edge will be central to the notion of ``bridges'' in Section~\ref{sec:typeANDcanon}).

Specifically, let $ab$ be a long non-inflection edge; let $P''$ be the part of \P up to $a$, and let $P'$ be the part of \P starting from $b$. For a vector $T$ let $P'(T)$ be the subpath $P'$ translated by $T$.
\begin{lemma}\label{lem:breakLong}For any $T$ there exists $\eps>0$ and a point $c$ such that the path $P''\-a\-c\-b'\-P'(\eps T)$ (i.e., the concatenation of $P''$, the 2-link path $acb'$, and $P'(\eps T)$) is a feasible discrete curvature-constrained path.\end{lemma}
\begin{proof}Assume that $ab$ is horizontal and that the path turns left at both $a$ and $b$ (so the edges adjacent to $ab$ are above it; Fig.~\ref{fig:breakLong}, left). If $T$ points up, we choose $c\in ab$ so that $|ac|=\ll$ (Fig.~\ref{fig:breakLong}, middle). The turn constraints and turn-over-length constraints are not broken because the sum of the turns at $c$ and $b'$ equals to the turn that the original path \P made at $b$. The length constraints are satisfied because even if $cb'$ is short it cannot be adjacent to a short edge (otherwise in the original path \P the long edge $ab$ was adjacent to a short edge -- a contradiction to Lemma~\ref{lem:ln}).

The proof for the case when $T$ points down is analogous: $ab$ is broken at the point $c$ at distance \ll from $b$, the normal part is translated by $T$, and the remainder is rotated to keep the connectivity (Fig.~\ref{fig:breakLong}, right).\begin{figure}\centering\includegraphics{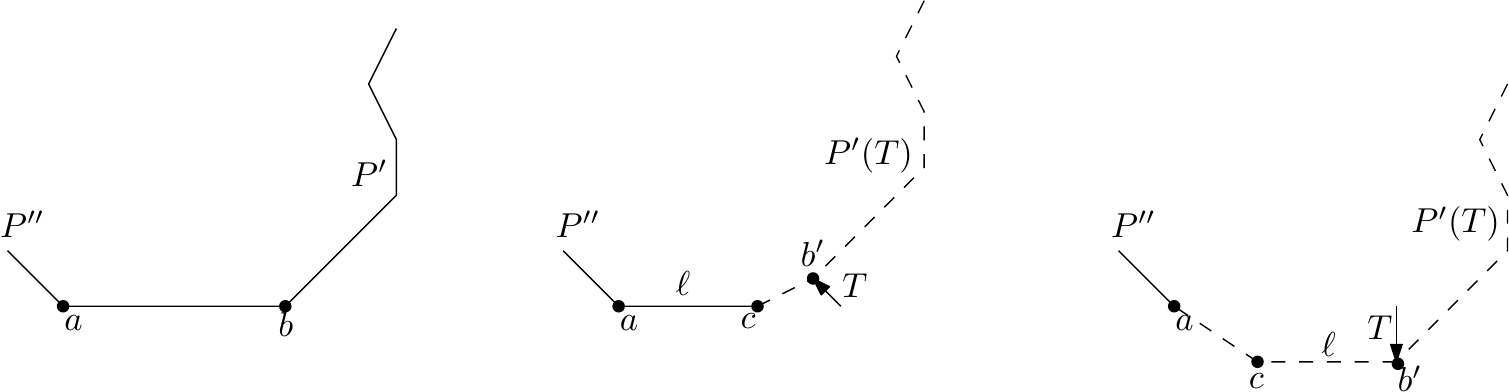}\caption{Left: The original path; $ab$ is horizontal. Middle: The modification if $T$ points up. Right: The modification if $T$ points down.}\label{fig:breakLong}\end{figure}\end{proof}
As immediate corollaries of Lemmas~\ref{lem:Ifreedom} and~\ref{lem:breakLong} we obtain that inflection and long edges have ``length freedom'':

\begin{lemma}\label{lem:IlengthFreedom}Let $ab$ be an inflection edge, and let $P'$ be the part of \P after $b$; for $\eps\in[0,1]$ let $\P(\eps)$ be the path obtained from \P by rigidly translating $P'$ towards $a$ along $ab$ by distance $\eps|ab|$. There exists $\eps>0$ such that $\P(\eps)$ is a feasible discrete curvature-constrained path.\end{lemma}
\begin{proof}The vector $\vec{ba}$ is admissible for $ab$.\end{proof}

\begin{lemma}\label{lem:Lfreedom}Let $ab$ be a long edge, and let $P'$ be the part of \P starting from $b$; for $\eps\in[0,1]$ let $\P(\eps)$ be the path obtained from \P by rigidly translating $P'$ towards $a$ along $sb$ by distance $\eps|ab|$. There exists $\eps>0$ such that $\P(\eps)$ is a feasible discrete curvature-constrained path.\end{lemma}
\begin{proof}Choose \eps small enough so that $|ab|-\eps|ab|>\ell$.\end{proof}
\subsection{Inflection and long edges are but a few}\label{sec:few}Using the local freedom of inflection and long edges, formalized above, we now show that in a \ddp there are only a couple of inflection+long edges. Say that inflection edges $ab$ and $cd$ have \e{similar turns} w.r.t.\ \P if the turn at $a$ is the same as the turn at $c$ --- both right of both left (and hence the turns at $b$ and $d$ are also the same --- both left or both right).
\begin{lemma}\label{lem:il}A \ddp cannot have
\begin{enumerate}
\item\label{item:2i}two inflection edges with similar turns, unless the edges are parallel
\item\label{item:li}a non-inflection long edge and an inflection edge
\item\label{item:2l}two long non-inflection edges
\end{enumerate}\end{lemma}
\begin{list}{}{}
\item[\e{Proof of~\ref{item:2i}}.]
    Let $P'$ be the part of \P between $b$ and $c$; since inflection edges with similar turns cannot be adjacent, $P'$ is non-empty. By the triangle inequality, each of the following two local modifications shortens \P (Fig.~\ref{fig:2i}): (1)~slide $P'$ along $cd$ moving $c$ towards $d$ and rotating $ab$ around $a$ so as to keep it connected to $b$; (2)~change the roles of $ab$ and $cd$, i.e., slide $P'$ along $ba$ moving $b$ towards $a$ and rotating $dc$ around $d$ so as to keep it connected to $c$. Moreover, one of these modification keeps the path feasible because either $\vec{cd}$ is admissible for $ab$, or $\vec{ba}$ is admissible for $dc$ (this is because of any two non-parallel vectors, one is to the left of the other).
\begin{figure}\centering\includegraphics{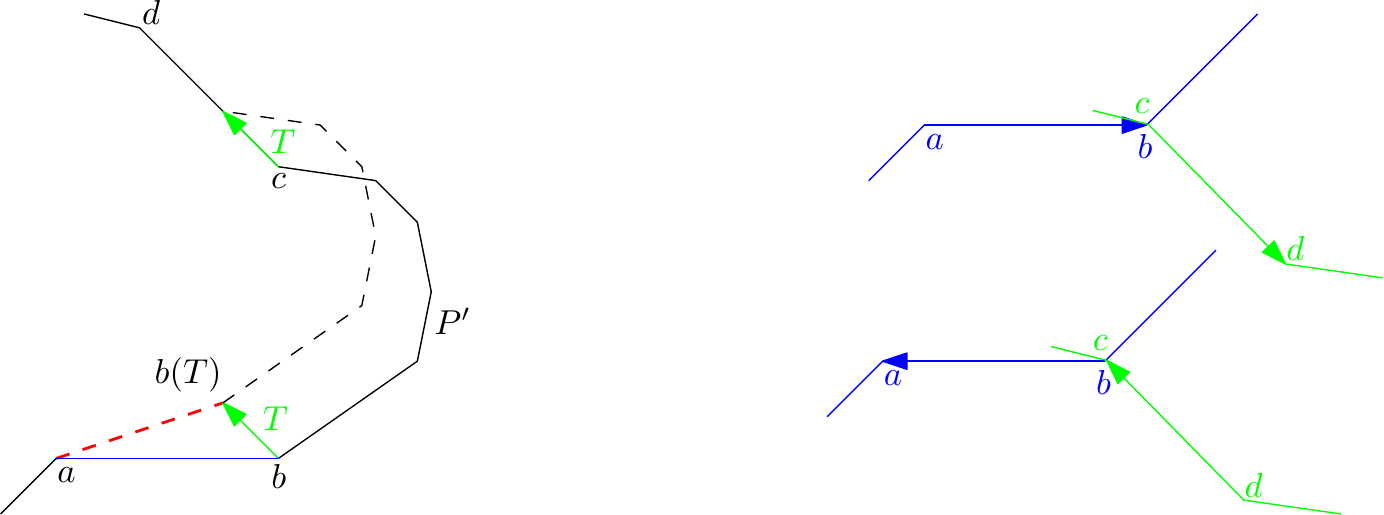}\caption{Left: $P'$ is moved along one edge while rotating the other; {\color{blue}{blue}}+{\color{green}{green}} is replaced by {\color{red}{red}}. Right: If $\vec{cd}$ is not admissible for $ab$, then $\vec{ba}$ is admissible for $dc$.}\label{fig:2i}\end{figure}
\item[\e{Proof of~\ref{item:li}}.]By Lemma~\ref{lem:infl}, the edges are not adjacent. Rigidly shift the subpath between the edges along the inflection edge (shortening the edge); a small shift is feasible due to Lemma~\ref{lem:IlengthFreedom}. Keep the path connected by breaking the long edge as appropriate (Lemma~\ref{lem:breakLong}). By the triangle inequality, this modification shortens the path (Fig.~\ref{fig:li}).\begin{figure}\centering\includegraphics{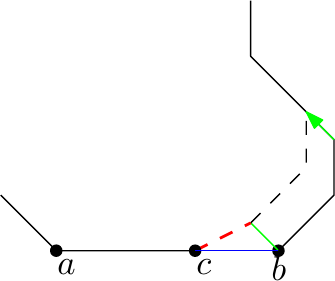}\caption{The modification replaces {\color{blue}{blue}}+{\color{green}{green}} by {\color{red}{red}}.}\label{fig:li}\end{figure}
\item[\e{Proof of~\ref{item:2l}}.]Analogous to~\ref{item:li}, using Lemma~\ref{lem:Lfreedom} in place of Lemma~\ref{lem:IlengthFreedom}.
\end{list}\qed

\section{Discrete arcs}\label{sec:arcs}The results from the previous section imply that, loosely speaking, pretty much all that is left for a \ddp is to have long chains of normal non-inflection edges, connected by few inflection and/or long and/or short edges. In this section we introduce ``discrete circular arcs'', which are essentially chains of normal edges with turns~\th.
\begin{definition}A \e{discrete circle} is a regular $\frac{2\pi}\th$-gon with side~\ll.\end{definition}
In what follows, we often omit the modifier ``discrete'' and call discrete circles just circles, discrete arcs just arcs, etc.

We define a discrete arc as a maximal subpath of \P that is also a subpath of a discrete circle; the definition is somewhat involved because formally, any length-at-most-\ll edge of \P is a subpath of a discrete circle -- and we want to avoid calling each and every short edge an arc:
\begin{definition}A \e{discrete circular arc} is a subpath $\pi=(p_1,\dots,p_m)$ of a curvature-constrained path \P, with the following properties:
\begin{enumerate}
\item $\pi$ is a subpath of some discrete circle;
\item $\pi$ is maximal in the sense that it is not contained in another subpath of \P which is also a subpath of some discrete circle;
\item if $m=2$, then at least one of $p_1,p_2$ is a vertex of \P, and in addition
\begin{enumerate}
\item if $p_1$ is a vertex of \P, but $p_2$ is not, then $p_1=u$ and the turn from the pre-edge to $p_1p_2$ is exactly \th (i.e., $\angle u'up_2=\th$);
\item if $p_2$ is a vertex of \P, but $p_1$ is not, then $p_2=v$ and the turn from $p_1p_2$ to the post-edge is exactly \th (i.e., $\angle p_1vv'=\th$);
\item if both $p_1,p_2$ are internal vertices of \P (i.e., $p_1\ne u,p_2\ne v$), then $|p_1p_2|=\ll$;
\item if both $p_1,p_2$ are vertices of \P and $p_1=u$, then the turn from the pre-edge to $p_1p_2$ is exactly \th or $|p_1p_2|=\ll$;
\item if both $p_1,p_2$ are vertices of \P and $p_2=v$, then the turn from $p_1p_2$ to the post-edge is exactly \th or $|p_1p_2|=\ll$.
\end{enumerate}\end{enumerate}
\end{definition}
Few remarks are
appropriate here (refer to Fig.~\ref{fig:example}):
\begin{itemize}
\item If $m=2$ and $p_1p_2$ is short, then one of $p_1,p_2$ is a terminal vertex of \P (e.g., in Fig.~\ref{fig:example} the short edge $ua$ is an arc since the turn onto it from the pre-edge is \th). If $m=2$ and $p_1,p_2$ are both internal vertices of \P (e.g., $g$ and $h$ in Fig.~\ref{fig:example}), then (by the maximality of the arc) the turn of \P at each of $p_1,p_2$ is less than~\th.
\item If $m=3$, then either $|p_1p_2|=\ll$ or $|p_2p_3|=\ll$ (since the arc is a subpath of a discrete circle, it has no long edges; no two short edges are adjacent).
\item If $m>3$, then the $m-3$ internal edges of an arc are normal, and the turn of \P at each internal vertex of the arc is exactly~\th.
\item Every vertex at which \P turns by $\th$, is part of an arc, and every edge incident to a turn-$\th$ vertex, will share a portion with an arc. If a short inflection edge makes angles \th with both adjacent edges, the short edge is a part of two arcs; that is, consecutive arcs along \P may overlap.
\end{itemize}
Note that only the internal vertices of an arc must be vertices of \P; the terminal vertices of an arc may lie in the middle of \P's edges (in Fig.~\ref{fig:example}, such are both terminal vertices $b,e$ of the arc $bc\-de$, as well as the terminal vertex $d$ of the arc $de\-f$). We use such vertices below when defining the type of a path, as well as in the local modifications by which we eliminate forbidden pairs of edges (we also used such a vertex earlier when identifying the possibility to break a long edge in Lemma~\ref{lem:breakLong}).

\subsection{Arcs and inflection}One important aspect in which discrete curvature-constrained paths differ form their smooth counterparts is that several consecutive discrete arcs may be similarly oriented (all going clockwise or all counterclockwise); in a smooth path adjacent arcs necessarily have a common inflection tangent. We show that, nevertheless, having two adjacent similarly-oriented arcs in a sequence of arcs is quite restrictive for a \ddp{} -- essentially, \e{all} arcs must then be similarly oriented (Lemma~\ref{lem:monot}). In addition, a path of more than 3 similarly oriented arcs may be either shortened or transformed to an equal-length path with only 2 arcs (Lemma~\ref{lem:circ}).

Specifically, for adjacent arcs $X,Y$ let $X\-Y$ denote the subpath going from the start of $X$ to the end of $Y$. The (long and technical) proof of the following lemma can be found in Appendix~\ref{app:lem:monot}:
\begin{lemma}\label{lem:monot}Let $X,Y,Z$ be consecutive arcs. If the subpath $X\-Y$ has no inflection edge, then neither does $Y\-Z$.\end{lemma}
Recall (Section~\ref{sec:freedom}) that a long edge has more freedom than an inflection edge (compare Lemmas~\ref{lem:Ifreedom} and Lemmas~\ref{lem:Lfreedom}); in a sense, a long edge
(which is capable of being broken)
has ``twice as much freedom'' as an inflection edge -- admissible vectors for the latter lie in a halfplane while for the former \e{any} vector is ``admissible''. In particular, the arguments in the proof of the above lemma can be carried out for long edges in place of inflection, leading to
\begin{lemma}\label{lem:monotL}Let $X,Y,Z$ be consecutive arcs. If the subpath $X\-Y$ has no inflection edge, then the subpath $Y\-Z$ does not have a long edge.\end{lemma}
We now show that a path consisting of 3 similarly oriented arcs may be transformed to an equal-length path ``with more \th-turns''. Say that a path is \e{flush} if the turn from the pre-edge onto the first edge of the path is exactly~\th. The (technical) proof of the next lemma can be found in Appendix~\ref{app:lem:circ}:
\begin{lemma}\label{lem:circ}A non-flush 
\ddp consisting of 3 consecutive arcs and having no inflection edge can be transformed to an equal-length path that is either flush, or consists of at most 2 arcs.\end{lemma}

\section{Path typification and canonization}\label{sec:typeANDcanon}Smooth Dubins paths can be classified as either arc-segment-arc or arc-arc-arc type. In this section, we develop a notion of type for a discrete curvature-constrained path, similar to the type of a smooth Dubins path, by labeling subpaths with $A$'s (discrete circular arcs) and $B$'s (bridges -- segments connecting consecutive arcs). Specifically, the type of a path is defined as follows:
\begin{enumerate}
\item Identify all arcs, and label each arc by an~$A$.
\item The unlabeled subpaths are now a set of straight segments disjoint other than possibly at endpoints; each such segment is called a \e{bridge} and labeled with a~$B$.
\end{enumerate}
The \e{length} of a type is the number of letters in it.

Remarks:
\begin{itemize}
\item We emphasize again that a bridge/arc may start or end in the middle of an edge of the path; i.e., the turn of \P at a bridge endpoint may be~0 (such is, e.g., the endpoint $b$ of the bridge $ab$ in Fig.~\ref{fig:example}). However, no edge of the path can support more than one bridge.
\item (As mentioned above,) consecutive arcs may overlap (e.g., arcs $bc\-de$ and $de\-f$ in Fig.~\ref{fig:example} have a common edge $de$). Bridges, on the contrary, have disjoint interiors, which are also disjoint from arcs.
\item The main feature of the bridge is the ``freedom to rotate locally'' (formalized in Lemma~\ref{lem:bridge}).
\item A type is uniquely defined.
\end{itemize}
The type is defined for arbitrary discrete curvature-constrained paths, not only for shortest such paths (\ddps). Since our ultimate goal is the characterization of the shortest paths, from now on we will consider only \ddps.

\subsection{Canonical form}In our development so far, we tacitly assumed that a path makes a non-zero turn at every vertex, i.e., that the path has no vertices in the middle of its edges. This is a usual and natural assumption in the treatment of polygonal paths, and we
lose
no generality by making it: any 0-turn vertex could be removed without violating the feasibility of the path -- removing such a vertex does not influence the path turns and does not decrease the edge lengths). On the contrary, \e{adding} vertices in the middle of edges may, in general, render a path infeasible (e.g., if a vertex is added on a normal edge, the edge splits into two adjacent short edges). Still, when proving that a path may have only one bridge, we would like to have vertices at certain points internal to edges. Specifically, we postulate that the endpoints of the bridges are \e{always} path vertices (by default, this may not be the case -- like, e.g., for point $b$ in Fig.~\ref{fig:example}); we need this to be able to ``break'' the path there (this is similar to how we broke long edges in Lemma~\ref{lem:breakLong}). That is, even if the endpoint of a bridge is in the middle of an edge, we add a vertex of the path there. The next lemma shows that such addition of vertices is feasible; the (straightforward) proof of the lemma appears in Appendix~\ref{app:lem:canonical}:
\begin{lemma}\label{lem:canoncial}Let $\P$ be a \ddp. The path $\P'$ obtained by adding to $\P$ as vertices all bridge endpoints is a \ddp.\end{lemma}
We say that the path $\P'$ is in \e{canonical} form. Where needed, we will assume that a \ddp is in canonical form, i.e., we will often jump between the original path \P and its canonical form (the above lemma allows us to do this). In particular, when necessary, we will assume that bridges are edges of the path, and can be rotated around their endpoints. Note that the during the canonization, vertices are appended only in the middle of long edges, which is similar to Lemma~\ref{lem:breakLong}.

We emphasize that we first do path typification (splitting the path into $A$s and $B$s), and only then -- canonization (adding $A$'s endpoints to vertices of the path). This is natural, since the type is defined for arbitrary paths, while the canonical form -- only for shortest paths.

One useful property of a path in canonical form is that it can be split at vertices into subpaths, each of which is discrete bounded-curvature path. If in addition, the pre-edge and post-edges are long or normal, 
then it is also a \ddps itself; we state this as a lemma without proof:
\begin{lemma}\label{lem:breakCanonical}Let $\pi$ the subpath between vertices $b,c$ of a \ddp in canonical form; let $a$ be the vertex preceding $b$, and $d$ be the vertex following $c$ (if $b=u$, then $a=u'$; if $c=v$, then $d=v'$). Then $\pi$ is a discrete bounded curvature path. If in addition $\vec{ab},\vec{cd}$ are long or normal, then $\pi$ is in fact a \ddp starting at $(b,\vec{ab} / |\vec{ab}|)$  and  ending at $(c,\vec{cd}/ |\vec{cd}|)$.\end{lemma}
Lemma~\ref{lem:breakCanonical} allows one in some cases to speak about a particular ``subtype'' --  a type subsequence, as a path in it's self. This is helpful when eliminating forbidden subsequence $AAAA$ (Section~\ref{sec:elimination}).
\subsection{``Freedom'' of bridges}We defined the bridges so that their vertices have certain freedom similar to (or, actually, exceeding) the freedom of inflection and long edges (Section~\ref{sec:freedom}). Specifically, let $ab$ be a bridge in a \ddp \P, and let $P'$ be the part of \P starting from $b$. For a vector $T$ let $\P(T)$ be the path obtained from \P by rigidly translating $P'$ by $T$, and connecting $a$ to $b'=b+T$ (the new position $b$) (Fig.~\ref{fig:Bfreedom}, left). The next lemma is proved in Appendix~\ref{app:lem:bridge}:
\begin{figure}\centering\includegraphics{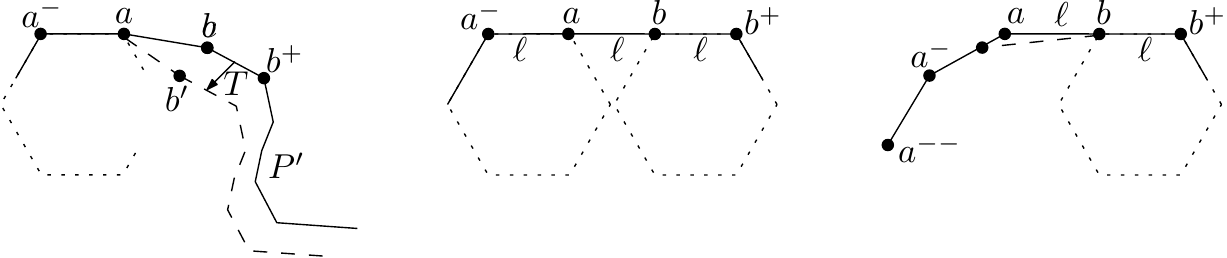}\caption{Left: The modification. The turn constraints are satisfied by definition of a bridge. The length and turn-over-length constraints can break only if $ab$ is normal, $ab'$ is short, and one of the edges adjacent to the bridge is short. Middle: Both $a$ and $b$ were added during canonization. Right: Only $b$ was added during the canonization: if $a^-a$ is short, a shortcut (dashed) is possible.}\label{fig:Bfreedom}\end{figure}
\begin{lemma}\label{lem:bridge}For any $T$ there exists $\eps>0$ such that $\P(\eps T)$ is a feasible discrete curvature-constrained path.\end{lemma}
\subsection{Bridges are but a few}We append bridges to the list of rare occurrences in a \ddp (the start of the list is in Lemma~\ref{lem:il}). Say that two bridges $ab,cd$ have \e{similar direction} w.r.t.\ \P if they are parallel and the vectors $\vec{ab},\vec{cd}$ point in the same direction. In Appendix~\ref{app:lem:b} we prove the following:
\begin{lemma}\label{lem:b}A \ddp may not have
\begin{enumerate}
\item\label{item:bi}a bridge and an inflection edge
\item\label{item:lb}a long edge and a bridge
\item\label{item:2b}two bridges, unless they have similar direction.
\end{enumerate}\end{lemma}

\section{Types of \ddps}\label{sec:typesDDP}We define a set of ``true types'' and a set of ``forbidden subtypes'' (Section~\ref{sec:trueANDforbidden}); using the facts that inflection edges, long edges and bridges are rare (Lemmas~\ref{lem:il} and~\ref{lem:b}) and that consecutive arcs must inflect (Lemmas~\ref{lem:monot} and~\ref{lem:circ}), we show that forbidden subtypes can never appear in a \ddp (Section~\ref{sec:elimination}). It will follow that for any start and end configurations there always exists a \ddp of a true type.
\subsection{True types and Forbidden subtypes}\label{sec:trueANDforbidden}We prove that every \ddp with \e{minimal-length} type is of the types $\cal{T}=\{B,A,AB,BA,AA,ABA,AAA\}$ -- the \e{true} types. This will be shown by considering the set $\cal{F}=\{BB,BAB,AAB,BAA,AAAA\}$ of \e{forbidden subsequences} and demonstrating that if the type has a member of \cal{F} as a subsequence, then either the path can be strictly shortened or modified to a same-length path with a shorter type. More specifically, we show that the forbidden subtypes $\cal{F}\setminus AAAA$ can \e{never} appear in a \ddp, while for the subtype $AAAA$ there always exists a true-type path with the same or smaller length. By symmetry we only need to prove that the subsequences $BB,BAB,AAB,AAAA$ can be excluded.
\subsection{Eliminating forbidden subsequences}\label{sec:elimination}
\subsubsection*{BB}Two adjacent bridges must belong to the same edge of \P (otherwise the bridges are not parallel -- a contradiction to Lemma~\ref{lem:b}(\ref{item:2b})), which is impossible by definition of a bridge.
\subsubsection*{BAB}By Lemma~\ref{lem:b}(\ref{item:2b}), the two bridges must have similar direction; this is possible only if there is an inflection edge between them, contradicting Lemma~\ref{lem:b}(\ref{item:bi}).
\subsubsection*{AAB}Let \ccone, \cctwo be the supporting circles of the two adjacent arcs. There are 3 ways of how the two arcs may connect: (1)~\ccone, \cctwo are flush, sharing a whole edge or part of an edge (Fig.~\ref{fig:AABee}); (2)~they share a single vertex (Fig.~\ref{fig:AABvv}); (3)~an edge of one circle
pivots
on a vertex of the other (Fig.~\ref{fig:AABve}). We consider the 3 cases separately.
\begin{list}{}{}
\item[\e{Edge-edge}]Let $e=\ccone\cap\cctwo$ be the common part of the two arcs. If $e$ is normal (i.e., \ccone, \cctwo share a whole edge), then $e$ is an inflection edge (for otherwise $\ccone=\cctwo$) -- a contradiction to Lemma~\ref{lem:b}(\ref{item:bi}). If $e$ is short, it is either a proper subset of a long edge of the original path or is an inflection edge (Fig.~\ref{fig:AABee}). However, having a long edge and a bridge contradicts Lemma~\ref{lem:b}(\ref{item:lb}), and having an inflection edge a bridge contradicts Lemma~\ref{lem:b}(\ref{item:bi}).
    \begin{figure}
    \begin{minipage}[c]{.45\columnwidth}\centering\includegraphics{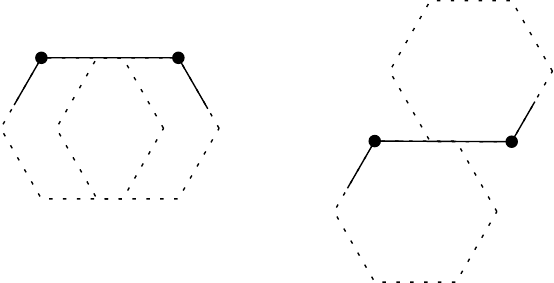}\caption{Edge-edge: the original path had a long edge.}\label{fig:AABee}\end{minipage}\hfill
    \begin{minipage}[c]{.45\columnwidth}\centering\includegraphics{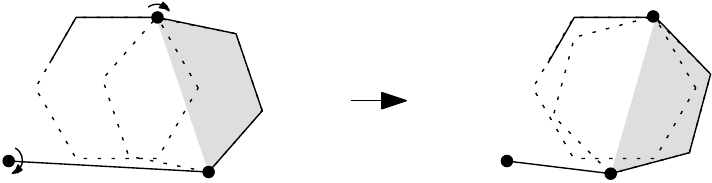}\caption{Vertex-vertex: rotate the second arc as a rigid body; rotate the bridge to keep the path connected.}\label{fig:AABvv}\end{minipage}
    \end{figure}
\item[\e{Vertex-vertex}]
It is easy to see that the second arc has the local freedom to rotate, in either direction, around the common vertex of the arcs without violating the feasibility of the path (Fig.~\ref{fig:AABvv}). Due to the rotation, the other endpoint of the arc may move slightly, which is no problem because it is bridge endpoint (Lemma~\ref{lem:bridge}). Rotating the arc in one of the directions decreases the length of the bridge (and hence the length of the path), unless the bridge belongs to the line through the arc's endpoints, in which case either the last edge of the arc is an inflection edge (contradiction to Lemma~\ref{lem:b}(\ref{item:bi})) or the bridge is inside \cctwo (contradiction to the turn constraint); refer to Fig.~\ref{fig:AABvv-aligned}.
\begin{figure}
\begin{minipage}[c]{0.4\columnwidth}\centering\includegraphics{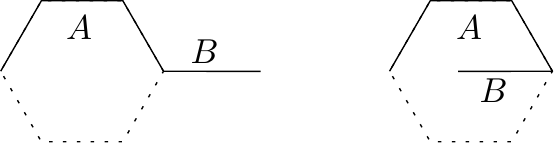}\caption{Bridge aligned with the arc endpoints.}\label{fig:AABvv-aligned}\end{minipage}\hfill
\begin{minipage}[c]{0.5\columnwidth}\centering\includegraphics{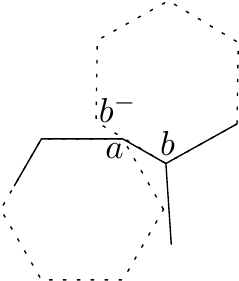}\caption{Vertex-edge: if the turn at $b$ is to the right, the turn over $ab$ is too sharp; if the turn at $b$ is to the left, then $ab$ is an inflection edge.}\label{fig:AABve}\end{minipage}
\end{figure}
\item[\e{Vertex-edge}]Suppose that an edge $b^-b$ of \cctwo rocks on a vertex $a$ of \ccone (Fig.~\ref{fig:AABve}). Then $ab$ is short (for otherwise, we are in the vertex-vertex or edge-edge case), and hence the second arc has more then one edge ($b$ cannot be a terminal vertex -- there is a bridge following the arc). But then $ab$ is an inflection edge (for otherwise, the turn-over-length constraint is violated for it -- the turn at $b$ is already \th). Having an inflection edge and a bridge contradicts Lemma~\ref{lem:b}(\ref{item:bi}).
\end{list}
\subsubsection*{AAAA}We show that a subpath consisting of 4 arcs $A_1A_2A_3A_4$ can either be shortened or transformed to an equal-length path with fewer arcs. Incidentally, \e{all} existing proofs of the structure of smooth Dubins paths go through proving non-optimality of $AAAA$ paths \cite[Lemma~2]{d-cmlca-57}, \cite[Lemma~11]{bcl-spbcp-94}, \cite[Lemma~26]{st-sprsc-91}, \cite[Lemma~7]{rs-opctg-90}; we are no exception. We do case analysis based on the number of inflection edges in an $AAAA$ path.
\paragraph{More than 2 inflection edges}In this case two inflection edges have similar turns -- a contradiction to Lemma~\ref{lem:il}(\ref{item:2i}).
\paragraph{Two inflection edges}Clearly, for any arcs $A_i,A_{i+1},i=1,2,3$, the subpath $A_i\-A_{i+1}$ can contain at most 1 inflection edge. Now, no matter which subpaths contain the two inflection edges, we can always apply Lemma~\ref{lem:monot} to reach a contradiction.
\paragraph{One inflection edge}If the subpath $A_1\-A_2$ has no inflection edge, then, by Lemma~\ref{lem:monot}, neither do $A_2\-A_3$ and $A_3\-A_4$, so there is no inflection edge at all. On the other hand, having the inflection edge in the subpath $A_1\-A_2$ (or, equivalently, in $A_2\-A_1$) contradicts Lemma~\ref{lem:monot} applied to the arcs $A_3,A_2,A_1$.
\paragraph{There are no inflection edges}Similarly to the AAB case above, there are 3 ways of how two consecutive arcs may connect: sharing (part of) an edge, sharing a single vertex, or with an edge of one arc pivoting on a vertex of the other (see Figs.~\ref{fig:AABee}, \ref{fig:AABvv} and~\ref{fig:AABve}). Again similarly to AAB, in the first case the original path (i.e., the path before the canonization) had a long edge, which constradicts Lemma~\ref{lem:monotL}. Still similarly to AAB, in the third case the turn from one arc to the next is too sharp. Thus, all three connections between the consectuive arcs in AAAA are vertex-vertex types, and the edges incident to the connection vertices are normal (Fig.~\ref{fig:AAAA}).
\begin{figure}\centering\includegraphics{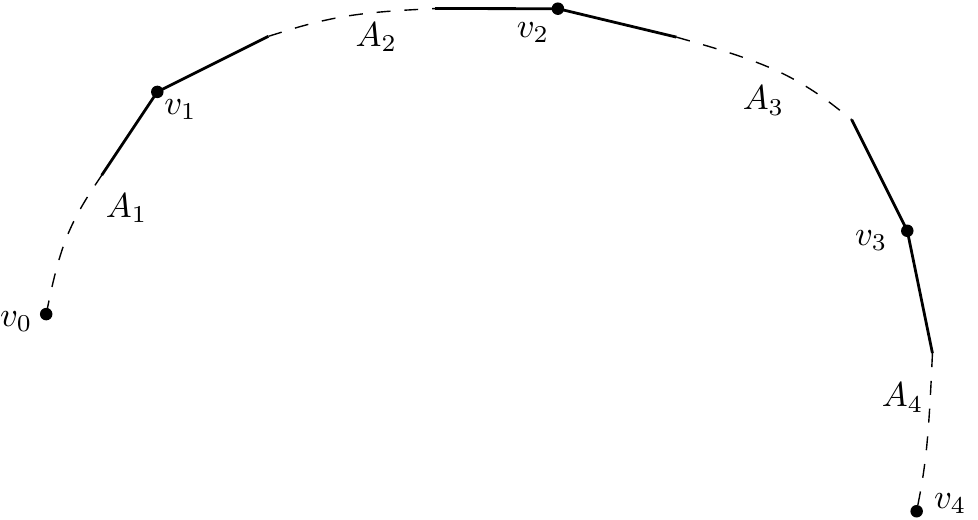}\caption{The arc $A_i,i=1\dots4$ goes from $v_{i-1}$ to $v_i$. The edges incident to $v_1,v_2,v_3$ are normal.}\label{fig:AAAA}\end{figure}

Consider now the subpath $\pi=v_1-\dots-v_4$ consisting of 3 arcs. The edge of $A_1$, incident to $v_1$ is normal; thus, by Lemma~\ref{lem:breakCanonical}, $\pi$ is a \ddp. The turn at $v_1$ is not \th (for otherwise $A_1$ and $A_2$ would have been the same arc); thus, $\pi$ is not flush. By Lemma~\ref{lem:circ}, either $\pi$ can be replaced by a same-length 2-arc path, or $\pi$ can be made flush (without changing length); in the latter case $A_1,A_2$ merge into a single arc. In both cases, the number of arcs in $\pi$ decreases to 2, and the $AAAA$ subpath becomes $AAA$.


\subsection{Structure of a \ddp}\label{sec:main}We are now ready to show that \ddps have the same structure as Dubins paths; in the next section we use this to obtain properties of the (smooth) Dubins paths as a limiting case.
\begin{theorem}\label{dccg}For any configurations \uu, \vv the collection of discrete curvature-constrained \uu-\vv paths contains a path of minimum length that consists of a discrete circular arc followed by a segment followed by a discrete circular arc, or is a sequence of at most 3 discrete circular arcs.\end{theorem}
\begin{proof}We only have to prove that the shortest path exists; by forbidden subsequence elimination (Section~\ref{sec:elimination}), the path must be of true type. (Note that the local shortenings that we exhibited above, do not prove existence of a shortest path; potentially it could
happen
that our local modifications continue ad infinitum, never reaching the minimum.)

For $n\in\mathbb N$ let $l_n$ be the infimum of the lengths of \uu-\vv paths with at most $n$ vertices (if no such path exists, $l_n=\infty$). A \uu-\vv path cannot be shorter than $l^*=\inf_{n\in\mathbb{N}}l_n$. We will first prove that there exists $n_0\in\mathbb{N}$ such that any path with more than $n_0$ vertices is strictly longer than $l^*$; this implies that $l^*=\inf_{n\le n_0}l_n=l_{n_0}$. We next prove that for any $n$, $l_n$ attains its minimum, i.e., that there exists a path with at most $n$ vertices whose length is exactly $l_n$. It will follow that there exists a path with length $l^*=l_{n_0}$.

Clearly, there exists at least one turn-constrained \uu-\vv path with finite length; let the length be $L$. Any path with $n$ vertices has at least $\lfloor \frac{n-3}{2} \rfloor$ long or normal edges (there are $n-1$ edges, at most 2 of them are terminal; the rest have at most one short per long or normal by Lemma~\ref{lem:ln}). Thus the length of the path is at least $\lfloor (n-3)/2 \rfloor \ll$. In particular, a path with more than $n_0=\lceil 2L / \ll \rceil + 3$ vertices will be strictly longer than~$L$.

Any path with at most $n$ vertices may be specified by a point in $(x_i,y_i)_{i=1}^n\in\mathbb{R}^{2n}$. The set of all paths of length at most $L$ is a compact subset of $\mathbb R^{2n}$. Let \cal{P} be the set of all discrete curvature-constrained \uu-\vv paths; since the length is a continuous function on \cal P, the existence of the shortest path will follow as soon as we can show that \cal P is closed. We do it by examining the constraints one-by-one:
\begin{list}{}{}
\item[\e{Turn constraints}]Each constraint can be written as a \e{non-strict} inequality; hence turn-constrained paths form a closed set.
\item[\e{Length constraints}]The condition that $i$th edge of the path is short defines an open set $E_i\subset\mathbb R^{2n}$; thus the constraint that the $i$th and the $(i+1)$st edges are not both short, excludes an open set $E_i\cap E_{i+1}$. Since there are finitely many (pairs of consecutive) edges, the length constraints for all of them exclude an open set, leaving the feasible set closed.
\item[\e{Turn-over-length constraints}]Let $T_i\subset E_i$ be the set of points for which the turn-over-length constraint is \e{not} satisfied for the $i$th edge; since $T_i$ is open, the set of infeasible points (w.r.t.\ the turn-over-length constraints) is open, meaning that its complement---the feasible set---is closed.
\end{list}
\end{proof}

\section{Smooth curvature-constrained paths}\label{sec:limit}From now on, we turn our attention to the smooth (``usual'', non-discrete) bounded-curvature paths. Following Dubins \cite{d-cmlca-57} seminal work, we will not require that paths have curvature defined at every point; instead, we consider paths whose \e{mean} curvature is bounded everywhere (see the first paragraph in Dubins paper for a discussion of this technicality). Specifically, let \g be a smooth (continuously differentiable) path, parameterized by its arclength. For any $t$ in the domain of \g, let $\g'(t)$ denote the derivative of \g at $t$ -- the unit vector tangent to \g at $t$. We require that the average curvature of \g is at most 1, which is formalized as follows: for any $t<s<t+\pi$ let $\a{s}{t}$ denote the angle between the directions of $\g'(s)$ and $\g'(t)$; the bound on the average curvature means that $\a{s}{t}\le s-t$. (In other words, $\g'$, viewed as mapping from the domain of \g to the unit circle, is 1-Lipschitz.) We call such paths \e{admissible}.

In this section we show how to discretize an admissible path into a discrete curvature-constrained path; the finer the discretization, the closer the discrete path is to the smooth one. This allows us to obtain the properties of shortest smooth paths as a limit of the properties of their discrete counterparts, which we established above (Theorem~\ref{dccg}). We discretize a smooth path by splitting it into 3 parts (with the first and the last part possibly empty). The middle part is discretized by choosing regularly spaced points, while the first and the last part are whatever remains from the regular discretization; the first and the last part are made equal-length. We prove that the resulting polygonal path has at most 2 short edges (the first and the last edge), and that it is a feasible discrete curvature-constrained path. The details follow.
\subsection{Preliminary lemmas}The correctness of our discretization hinges on the next two lemmas.
\begin{lemma}\label{lem:length}For any $t$ and $t<s<t+\pi$ in the domain of \g, $|\g(s)-\g(t)|\ge2\sin\frac{s-t}2$.\end{lemma}
\begin{proof}Assume w.l.o.g.\ that $t=0$ and that $\g(t)$ is at the origin ($\g(0)=(0,0)$); the lemma is then equivalent to $|\g(s)|\ge2\sin\frac s2$ (Fig.~\ref{fig:angle}, left). We prove this by lower-bounding the derivative of $|\g(s)|^2$:
\[(|\g(s)|^2)'=2\g(s)\cdot\g'(s)=2\int_0^s\g'(\tau)\cdot\g'(s)\,d\tau=2\int_0^s\cos\a{s}{\tau}\,d\tau\ge2\int_0^s\cos(s-\tau)\,d\tau=2\sin s\]
Hence $|\g(s)|^2\ge2(1-\cos s)=4\sin^2\frac s2$.\end{proof}
\begin{figure}\centering\includegraphics{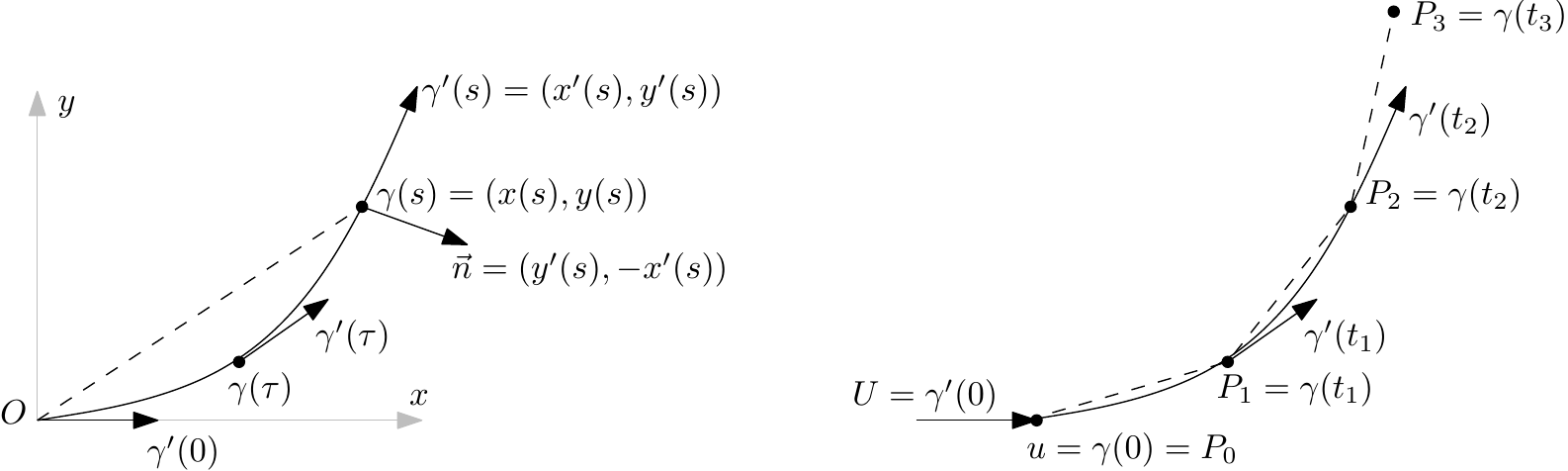}\caption{Left: $|O\g(s)|\ge2\sin\frac s2$, the slope of $O\g(s)$ is at most $\tan\frac s2$. Right: $P_1P_2$ and $P_2P_3$ are non-short. The turn from $P_1P_2$ onto $\g'(t_2)$ is at most \th/2; such is also the turn from $\g'(t_2)$ onto $P_2P_3$. If $\delta>0$, then $t_1=\delta/2<\th/2$, and $P_0P_1$ is short; in this case, the angles between $P_0P_1$ and the tangents to \g at $P_0,P_1$ are at most $\th/4$ each.}\label{fig:angle}\end{figure}
\begin{lemma}\label{lem:angle}The angle between $\g'(t)$ and the ray $\g(t)\g(s)$ is at most $\frac{s-t}2$.\end{lemma}
\begin{proof}Assume again that $t=0$ and that $\g(0)=O$; also assume w.l.o.g.\ that $\g'(0)$ is horizontal ($\g'(0)=(1,0)$). Let $\g(s)=(x(s),y(s))$, and let $k(s)=y(s)/x(s)$ be the slope of the ray $O\g(s)$ (Fig.~\ref{fig:angle}, left). Then the lemma is equivalent to $k(s)\le\tan\frac s2$, which we will prove by showing that $k'\le\frac{1-\cos s}{\sin^2s}=\frac1{2\cos^2(s/2)}=(\tan\frac s2)'$.

First of all, for any $s$ we have $x'(s)=(x'(s),y'(s))\cdot(1,0)=\g'(s)\cdot\g'(0)=\cos\a{s}{0}\ge\cos s$; thus, $x(s)\ge\sin s$. Next, consider the unit vector $\vec n=(y'(s),-x'(s))$, orthogonal to $\g'(s)$ (Fig.~\ref{fig:angle}, left). By definition, for any $\tau<s$ the angle between $\g'(\tau)$ and $\g'(s)$ is at most $s-\tau$; hence, the angle between $\g'(\tau)$ and $\vec n$ is at least $\pi/2-(s-\tau)$, from whence $(x'(\tau),y'(\tau))\cdot\vec n\le\cos(\pi/2-(s-\tau))$, or $x'(\tau)y'(s)-y'(\tau)x'(s)\le\sin(s-\tau)$. Integrating over $\tau$ from 0 to $s$, we get $x(s)y'(s)-y(s)x'(s)\le1-\cos s$. Combining this with $x(s)\ge\sin s$, we obtain what we need: $k'=(y/x)'=\frac{y'x-x'y}{x^2}\le\frac{1-\cos s}{\sin^2 s}$.
\end{proof}
\subsection{Discretization}Let $\th<|\gamma|$ be a number, and suppose that $|\gamma|=m\th+\delta$ for some $m\in\mathbb{N},\delta<\th$ (i.e., $m=\lfloor|\gamma|/\th\rfloor$ and $\delta$ is the remainder).
\begin{definition}\label{def:ndisc}The \th-discretization of $\gamma$ is a polygonal path $P$ with vertices $P_i=\gamma(t_i),i=0\dots k$ such that 
\begin{itemize}
\item $t_0=0$, $t_k=|\gamma|$.
\item If $\delta=0$, then $k=m=|\gamma|/\th$, and $t_i-t_{i-1} = \th$ for all $i=1\dots m$.
\item If $\delta>0$, then $k=m+2$, $t_i-t_{i-1} = \th$ for $i=1 \dots m+1$, and $t_1-t_0=t_{m+2}-t_{m+1}=\delta / 2$.
\end{itemize}
\end{definition}
Let $\ll=2\sin\frac\th2$.
By Lemma~\ref{lem:length}, whenever $t_i-t_{i-1}=\th$, we have that $P_iP_{i-1}\ge\ll$. Thus, all edges of $P$, except possibly for the first and the last one, are non-short. If the first and the last edges are short, the path has at least 3 edges, so the two short edges are non-adjacent. Thus, length constraints hold for~$P$.

We now examine the turns of $P$. Instead of directly looking at the angles between edges of $P$, we look at the angle that an edge makes at its endpoints with tangents to \g; we bound these angles using Lemma~\ref{lem:angle}. Specifically, let $P_{i-1}P_i,P_iP_{i+1}$ be two non-short edges. By Lemma~\ref{lem:angle}, the angle between $P_{i-1}P_i$ and the tangent $\g'(t_i)$ to \g at $P(i)$ is at most \th/2 (Fig.~\ref{fig:angle}, right); symmetrically, the angle between the tangent and $P_iP_{i+1}$ is also at most \th/2. Thus, the turn of $P$ at $P_i$ is at most \th, and the angle constraint is satisfied. Next, if $P_0P_1$ is a short edge, then $t_1-t_0<\delta/2<\th/2$, and by Lemma~\ref{lem:length}, the turn from $\g'(0)$ onto the edge is at most \th/4. Similarly, the turn from $P_0P_1$ to $\g'(t_1)$ is at most \th/4; finally the turn from $\g'(t_1)$ onto $P_1P_2$ is at most \th/2. Thus, overall, the turn of $P$ from $\g'(0)$ onto the first non-short edge, $P_1P_2$, is at most $\th/4+\th/4+\th/2=\th$.

Let $\cal U = (\gamma(0),\gamma(0)'),\cal V = (\gamma(|\gamma|),\gamma(|\gamma|)')$. It follows from the above that $P$ is a $\cal U\-\cal V$ discrete curvature-constrained path with parameters $\th$ and $\ll$, i.e., a path whose turns are constrained by $\th$ and lengths -- by $\ll$ (recall Definition~\ref{def:dccp} for the exact meaning of the parameters in the constraints).

Finally, for $n\in\mathbb N$ let $l_n=2\sin\frac\pi n$ be the side length of the regular $n$-gon inscribed in unit circle; let $\th_n=2\pi/n$. Suppose that $n$ is large enough so that $\th_n<|\gamma|$, 
and let $P_n$ be the $\th_n$-discretization of $\gamma$. The preceding discussion implies the following:
\begin{lemma}\label{lem:dis}$P_n$ is a \uu-\vv discrete curvature-constrained path with parameters $\th_n$ and $l_n$.\end{lemma}
\subsection{Dubins paths properties}We are now ready to give new proofs of Dubins result. Let \uu, \vv be two arbitrary configurations, and let $\mathbb C$ be the collection of \uu-\vv paths with bounded mean curvature. Say that a path $\gamma\in\mathbb C$ is \e{Dubins-type} if it consists of a circular arc followed by a segment followed by a circular arc, or is a sequence of at most 3 circular arcs.
\begin{theorem}[\cite{d-cmlca-57}]\label{dubins}$\mathbb C$ contains a Dubins-type path of minimum length.\end{theorem}
\begin{proof}Let $\tau_n$ be a \uu-\vv \ddp with parameters $\th_n$ and $l_n$. By Theorem~\ref{dccg} there are only finitely many types of \ddps; thus, we may assume, possibly changing to a subsequence, that all paths $\tau_n$ have the same type. Suppose, for instance that $\tau_n$ are of type $AAA$ (the other cases are similar). 
As $n\rightarrow\infty$ each of the discrete arcs converges to a circular arc.
Denote the limit curve, consisting of the 3 circular arcs, by $\tau^*$.

Let $\g\ne\tau^*$ be another \uu-\vv path with bounded mean curvature. For the sequence $P_n$ of $\th_n$-discretizations of \g we have $|\tau_n|\le|P_n|\le|\g|$. Taking the limit, we obtain $|\tau^*|\le|\g|$, which proves that $\tau^*$ is a curve of minimum length.\end{proof}


\section{Conclusion}\label{sec:discussion}We studied a discrete model of curvature-constrained motion. We chose one particular way of defining the discrete motion; many other versions are possible. Also, even within our framework, the definition of \ddp may be
modified
in many ways (the pre- and post-edges $\u\u',\v\v'$ may be mandated to be edges of \P, all edges may be required to have length at least \ll, the turn constraints may be imposed differently, etc.); some ways might lead to more natural (and possibly, shorter) definitions than ours. Our choice of the definition was prompted merely by technical details -- we found it easier to allow short edges etc.; since short edges disappear in the limit $\ll\rightarrow0$, the length and turn constraints at terminal edges are not a significant factor.

In any case, we admit that there may exist other definitions of discrete curvature-constrained motion; perhaps, the only objective requirements for the model are that smooth paths can be discretized, and that the model contains smooth paths as a limiting case of the discretization. This allows one to obtain results concerning curvature-constrained motion via the limiting argument similar to ours. As one example of such a result we mention the theorem that a curvature-constrained path must have length at least $\pi$ in order to be able to enter the interior of the unit circles tangent to the path's initial configuration. The theorem was proved using ``continuous'' methods in several prior works (\cite[Proposition~6]{d-cmlca-57}, \cite{m-pmcl-61,g-pdeppc-64,r-ca-70}, \cite[Lemma~2]{gunt-cfdp-04}); the discretize-and-take-the-limit technique may give yet another proof (of course, we do not claim to have invented the technique in this paper -- it has been used for centuries, e.g., for isoperimetric problems.)

\paragraph{Acknowledgements}We thank Sergey Bereg, Stefan Foldes, Irina Kostitsyna and Joe Mitchell for discussions.
{{{\bibliographystyle{abbrv}\bibliography{awkward}}}
\appendix
\section*{Appendix}
\section{Flight paths}\label{app:flightpaths}
\begin{figure}[h!]\centering\includegraphics[width=.8\columnwidth]{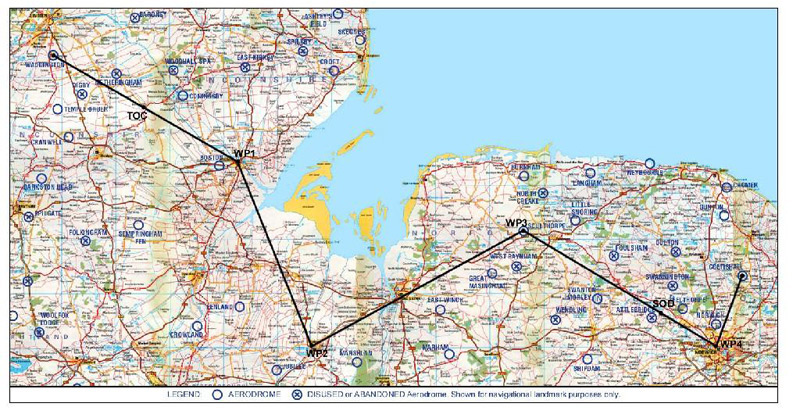}\\\includegraphics[width=.35\columnwidth]{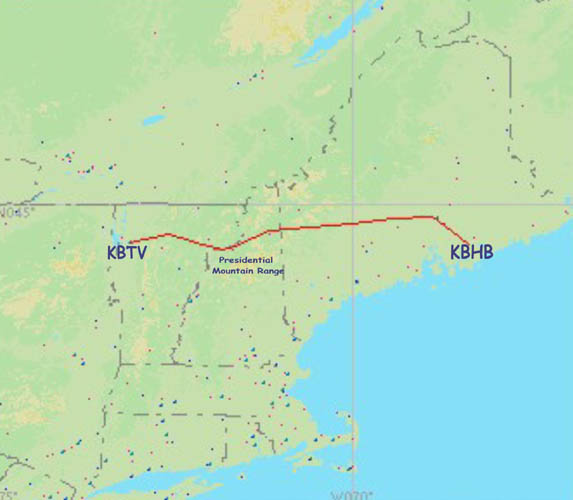}\quad\includegraphics[width=.35\columnwidth]{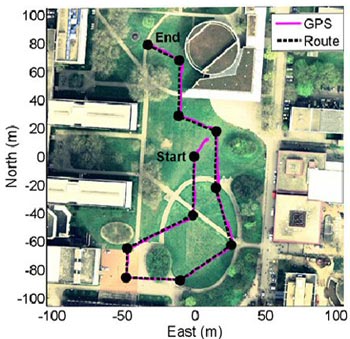}
\caption{Images from avsim.com, flightsim.com, spyzone.com}\label{fig:flighplans}\end{figure}

\section{Proof of Lemma~\ref{lem:monot}}\label{app:lem:monot}
If $X$ and $Y$ are ``flush'', i.e., if the last edge of $Y$ has the same supporting line as the first edge of $Y$, then the connection between the arcs is a long edge (Fig.~\ref{fig:monot}, left). In this case, Lemma~\ref{lem:il}(\ref{item:li}) implies that there can be no inflection edge anywhere in the path, in particular -- in $Y\-Z$, so we are done. Thus we will assume that $X$ and $Y$ intersect in a vertex (Fig.~\ref{fig:monot}, middle).

Let $b=X\cap Y$ be the common vertex of $X$ and $Y$, and suppose that $Y\-Z$ has an inflection edge $cd$. Since $b$ is the first vertex of $Y\-Z$, and an inflection edge in a subpath cannot be its first edge (an inflection edge is an internal edge), we have that $b\ne c$. Let $P_{bc}$ denote the subpath between $b$ and~$c$. Since arcs do not have inflection edges, $P_{bc}$ is a subpath of $Y$ (possibly, $P_{bc}=Y$ --- if $cd$ is not an edge of $Y$; in that case, however, $cd$ must be an edge of $Z$, like in Fig.~\ref{fig:monot}, middle).
\begin{figure}\centering\includegraphics{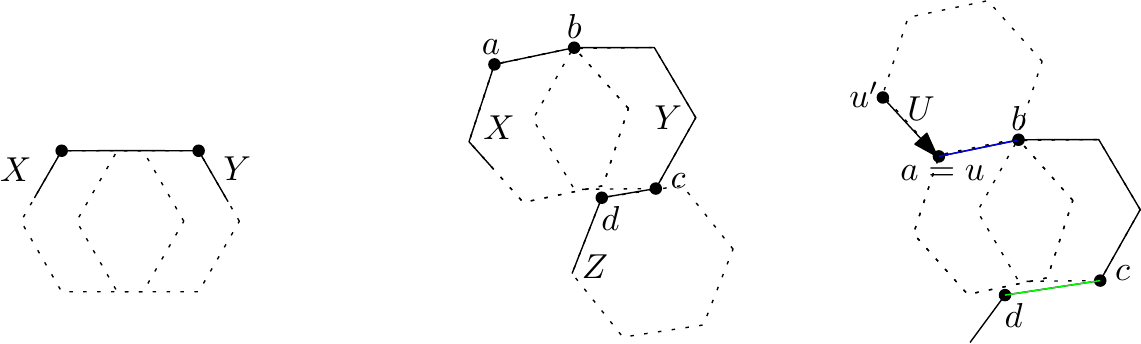}\caption{Left: If $X,Y$ are flush, there is a long edge. Middle: $b=X\cap Y$. $b\ne c$. $c$ is a vertex of $Y$, $d$ is a vertex of $Z$. $cd$ is aligned either with an edge of $Y$ or with an edge of $Z$ (or both --- $Y$ and $Z$ are flush). Right: If $a$ has no angular freedom to rotate \ccw, then $ab$ is the second inflection edge in the path $u'u\textrm-\P$; apply the argument of Lemma~\ref{lem:il}(\ref{item:2i}) to the path.}\label{fig:monot}\end{figure}

Let $a$ be the vertex (of $X$) preceding $b$, and let $P_{ac}$ denote the subpath between $a$ and $c$. Suppose that at $b$ (and hence also at all internal vertices of $X\-Y$) the path turns to the right. We would like to use an angular freedom at $a$ to rotate $P_{ac}$ \ccw around $a$. This is always possible, except when $a=u$ and the turn from the pre-edge onto $ab$ is to the left and is exactly \th (Fig.~\ref{fig:monot}, right). In the latter case, however, we can shorten \P by sliding $P_{bc}$ either along $ba$ or along $cd$ as in the proof of Lemma~\ref{lem:il}(\ref{item:2i}) -- this is, essentially, an application of the lemma to the path obtained by appending the pre-edge to \P (the appended path has 2 inflection edges -- $ab$ and $cd$). Thus in what follows we will assume that $P_{ac}$ can be rotated around~$a$.

\begin{wrapfigure}{r}{.2\columnwidth}\centering\includegraphics{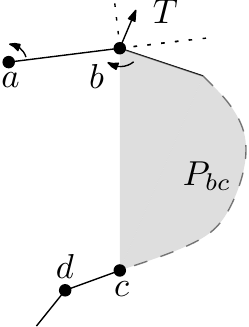}\end{wrapfigure}
Assume that $c$ is at the origin and that $cb$ is vertical and points up. We will do different local modifications depending on the direction of $cd$. All modifications will move $P_{bc}$ rigidly -- either rotating it around $b$ or translating. The rotation is no problem at $b$ -- there is a rotational freedom in any direction. In translating $P_{bc}$ we will ensure that the translation vector $T$, when applied to $b$, is to the left of the directed line $ab$ and makes an obtuse angle with $ab$ -- this way $ab$ will be rotated \ccw and also will increase length (so that the length and turn-over-length constraints continue hold). Some of our modifications rotate $ac$ \ccw about~$a$, and one rotates $P_{ac}$ \ccw about~$a$.

First of all, since $X\-Y$ is a right-turning path, both $ab$ and $cd$ are to the right of $cb$; thus, $cd$ cannot live in the first or fourth quadrant. If $cd$ is in the third quadrant, we simply rotate $P_{bc}$ \cw. This will locally move $c$ in the $-x$ direction, which is admissible for $dc$ (Definition~\ref{def:admissible}) and also shortens the edge. So from now on we assume that $cd$ is in the second quadrant. We separately consider the cases of $abc$ being non-obtuse and obtuse.

If $\angle abc$ is non-obtuse (Fig.~\ref{fig:monotModifIIacute}), drop the perpendicular $cc'$ from $c$ to (the supporting line of) $ab$. There are two possibilities of where $cd$ lies w.r.t.\ $cc'$:
\begin{list}{}{}
\item[\e{$cd$ belongs to the (closure of the) wedge $bcc'$ (Fig.~\ref{fig:monotModifIIacute}, left)}]Slide $P_{bc}$ along $cd$ towards $d$ by a small vector $T$. Since $cd$ is to the right of $cc'$, we have that $bb(T)$ is to the right of $bb'$, where $b(T)=b+T$ and $bb'$ is parallel to $cc'$ (and hence perpendicular to $ab$); thus, the modification does not shorten $ab$, and the length and turn-over-length constraints are satisfied. Since $cd$ is to the left of $cb$, $bb(T)$ is to the left of the vertical ray starting at $b$, and hence $abb(T)$ has a left turn at $b$, which means that the turn at $a$ decreases, so the turn constraints are also satisfied for the modified path.
\item[\e{$cd$ is strictly outside the wedge $bcc'$ (Fig.~\ref{fig:monotModifIIacute}, right)}]Rotate $ab$ about $a$ \ccw, translating $P_{bc}$ (parallel to itself, without rotation) -- locally this means that $P_{bc}$ is translated in the direction $cc'$, which is admissible for $cd$ by the assumption that $cc'$ is to the right of $cd$. To see that the length of $cd$ decreases note that the angle $\angle dcc'$ is acute -- this is because both $cd$ and $cc'$ live in the same (the third) quadrant.
\end{list}
\begin{figure}\centering\includegraphics{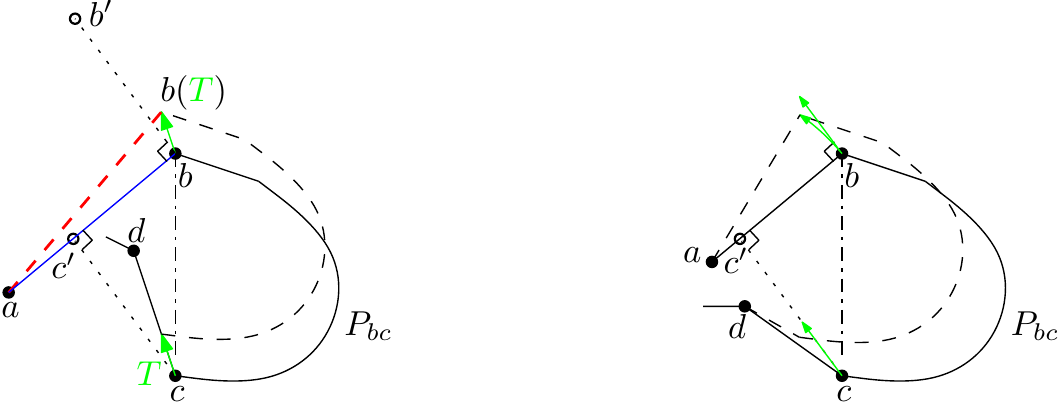}\caption{$cc'\bot ab$. (The figure is schematic: e.g., $ab$ is too long.) If $\angle bcd\le\angle bcc'$ (left), take $T=\eps\vec{cd}$; the modification replaces {\color{blue}{blue}}+{\color{green}{green}} by {\color{red}{red}}. If $\angle bcd>\angle bcc'$ (right), rotate $ba$ around~$a$; this locally moves $P_{bc}$ in the direction $cc'$ -- admissible for $cd$ and shortening $cd$.}\label{fig:monotModifIIacute}\end{figure}

If $\angle abc$ is obtuse (Fig.~\ref{fig:monotModifIIobtuse}), draw the ray $cc^*$ parallel to $ab$. Again, there are two possibilities of where $cd$ lies w.r.t.\ $cc^*$:
\begin{list}{}{}
\item[\e{$cd$ belongs to the (closure of the) wedge $bcc^*$ (Fig.~\ref{fig:monotModifIIobtuse}, left)}]Rotate $ab$ about $a$ \ccw, translating $P_{bc}$ (parallel to itself, without rotation) -- locally this means that $P_{bc}$ is translated in the direction $cc'$ normal to $ab$, which is admissible for $cd$ because $cd$ is to the left of $cb$ (since $cd$ is in the third quadrant) while $cc'$ is to the right of $bc$ (since $\angle abc$ is obtuse). To see that the length of $cd$ decreases note that the angle $\angle dcc'$ is acute -- this is because $cd$ is to the right of $cc^*$ (since $\angle bcd\le\angle bcc^*$).
\item[\e{$cd$ is strictly outside the wedge $bcc^*$ (Fig.~\ref{fig:monotModifIIacute}, right)}]This is the hardest case as no single transformation from the usual suspects---rotation around $a,b$, shift along $cd$---works alone (either the path becomes infeasible or lengthens). We thus carefully apply two transformations in succession.

Specifically, let $l$ be the supporting line of $cd$; we fix $l$ and \e{do not} move it as $cd$ changes during our transformations (in fact, after \e{both} our transformations are applied, $cd$ gets back aligned with $l$). Our first transformation is a rotation of $P_{bc}$ \cw around $b$ by a small angle \pb; since we rotate \cw, the turn from any edge of $P_{bc}$ to $l$ decreases by \pb, and since $cd$ is in the second quadrant, the rotation moves $c$ below $l$. (Note the difference: the path's turn angle at $c$ \e{increases} by some unspecified amount -- the turn direction is not admissible for $cd$; still, the turn onto the line $l$, which stays put and is not moved with the modification, \e{decreases} by exactly \pb.) Our second transformation is a rotation of $P_{ac}$ \ccw around $a$; we rotate until $c$ comes up to $l$. Let \pa be the angle of the rotation.

Let us see what happened with the path after the two transformations. The turns at all vertices except $a,b,c$ stayed the same. At $a$ and $b$ we had the freedom, so the turn constraints there are not violated; at $c$ the turn decreased by $\pb-\pa$. Claim~\ref{claim:2rotations} below assures that $\pb>\pa$, and so the turn constraints are satisfied. Now, the only edge that changed the length is $cd$, but it is an inflection edge, so the turn-over-length constrains are irrelevant for it, and the length constraints are satisfied because an inflection edge is adjacent to normal edges (Lemma~\ref{lem:infl}). Thus, the transformed path stays feasible. In addition, Claim~\ref{claim:2rotations} shows that our mission is complete -- the transformations make $cd$ shorter.
\end{list}
\begin{figure}\centering\includegraphics{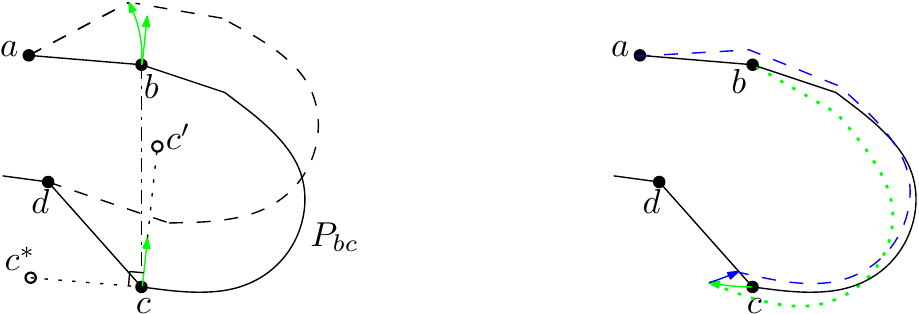}\caption{$cc^*\||ab$. If $\angle bcd\le\angle bcc^*$ (left), rotate $ba$ around~$a$; this locally moves $P_{bc}$ in the direction $cc'$ -- admissible for $cd$ and shortening $cd$. If $\angle bcd>\angle bcc'$ (right), first rotate $P_{bc}$ around $b$ and then rotate $P_{ac}$ \ccw around $a$ restoring the direction of $cd$.}\label{fig:monotModifIIobtuse}\end{figure}
This finishes proof Lemma~\ref{lem:monot}.\qed
\begin{claim}\label{claim:2rotations}$|c_2d|<|cd|$, where $c_2$ is the final position of $c$. Also, $\pb>\pa$.\end{claim}
\begin{proof}Refer to Fig.~\ref{fig:2rotations}. Let $a'$ be the foot of the perpendicular dropped onto $cd$ form $a$. We claim that $a'$ is to the left of $c$. Indeed, let $b'$ be the foot of the perpendicular dropped onto $cd$ form $b$. Since $a$ is to the left of $b$, $a'$ is to the left of $b'$. But $b'$ is to the left of $c$ because $\angle abc>\frac\pi2$.
\begin{figure}\centering\includegraphics{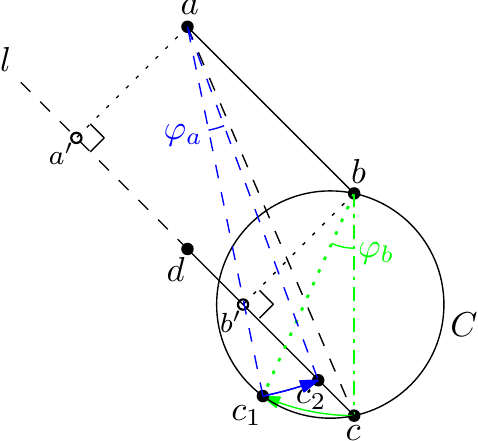}\caption{$a'$ is to the left of $b'$ to the left of~$c$. $\angle abc>\angle abc_1\Longrightarrow|ac|>|ac_1|=|ac_2|\Longrightarrow|c_2d|<|cd|$. $\pa=\angle c_1ac_2<\angle c_1ac<\angle c_1bc=\pb$.}\label{fig:2rotations}\end{figure}

Let $c_1,c_2$ be the images of $c$ after the first and the second transformation resp. For the triangles $abc,abc_1$ we have $|bc|=|bc_1|,\angle abc_1=\angle abc-\pb$, so by the cosine theorem $|ac_1|<|ac|$. Since $|ac_2|=|ac_1|$, we have $|ac_2|<|ac|$. Since the foot of the perpendicular from $a$ onto $cd$ lies to the left of $c$, we have that $c_2$ is closer to $d$ than $c$ (proving the first statement of the claim) and also that $\pa=\angle c_1ac_2<\angle c_1ac$.

Let $C$ be the circle through $b,c,c_1$. By continuity, for small enough \pb, $a$ is outside $C$ --- for $\pb=0$, the circle is the diametrical circle of $bc$, and $a$ is strictly outside it because $\angle bcd>\angle bcc^*$ (where $cc^*$ is the line parallel to $ab$ --- recall that we are still in this case). Thus, $\angle c_1ac<\angle c_1bc=\pb$.\end{proof}

\section{Proof of Lemma \ref{lem:circ}}\label{app:lem:circ}
Let \P be the path, and $X,Y,Z$ be the arcs of \P (Fig.~\ref{fig:circ}, left). Suppose that \P turns to the right at its internal vertices. If the turn from the last edge of $Z$ onto the post-edge is to the left, then \P augmented with the post-edge has an inflection edge (the last edge of $Z$), and we can shorten \P as in the proof of Lemma~\ref{lem:monot} --- this is, essentially, an application of the lemma to the augmented path consisting of $Y,Z$, and the post-edge treated as a separate arc. Thus, we will assume that the turn onto the post-edge is to the right; similarly, the turn from the pre-edge is to the right. By assumption the turn at $u$ is strictly less than $\th$.
\begin{figure}\centering\includegraphics{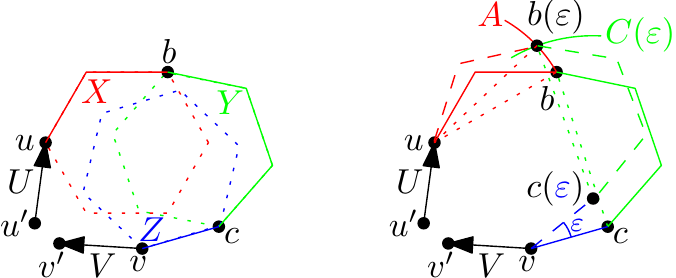}\caption{Left: The whole $u'\textrm-v'$ path is right-turning. Right: $|c(\eps)b(\eps)|=|cb|,|ub(\eps)|=|ub|$.}\label{fig:circ}\end{figure}

By Lemma~\ref{lem:monotL}, \P does not have long edges; thus the consecutive arcs intersect at vertices
(cf.\ Fig.~\ref{fig:circ}, left).
Let $b=X\cap Y,c=Y\cap Z$. We modify \P by a sequence of 3 rotations using the angular freedom at $u,b,c,v$. Specifically, for a real number \eps let $c(\eps)$ denote the image of $c$ after rotation by \eps around $v$, and let $C(\eps)$ be the circle of radius $cb$ the centered at $c(\eps)$ (Fig.~\ref{fig:circ}, right). Let $A$ be the circle of radius $|ub|$ centered at $u$, and let $b(\eps)=A\cap C(\eps)$. Since $u,b,c$ are not collinear, the circles $A$ and $C(0)$ intersect ``properly'' (formally, there is another point of intersection in addition to $b=b(0)$; we set $b(\eps)$ to be the point of the intersection closer to $b$). This means that there is some positive-length interval \cal{E} containing 0 in the interior, such that for any $\eps\in\cal{E}$ we have $A\cap C(\eps)\ne\emptyset$.

For $\eps\in\cal{E}$ let $\P(\eps)=X(\eps)Y(\eps)Z(\eps)$ be the path consisting of 3 arcs: $u\-b(\eps)$ arc $X(\eps)$, $b(\eps)\-c(\eps)$ arc $Y(\eps)$, and $c(\eps)\-v$ arc $X(\eps)$; the arcs are perturbed versions of $X,Y,Z$ and the length of $P(\eps)$ equals that of \P. Because of the non-zero freedom at each of $u,b,c,v$, there exists a positive-length interval $\cal{E}'\subset\cal{E}$ such that for any \eps in $\cal{E}'$ the path $P(\eps)$ is a feasible path. Indeed, since our perturbation does not change lengths, the only reason why \P may become infeasible is because of a turn constraint. But \P had non-zero freedom at $u,b,c$ \e{both} in \cw and \ccw directions -- by our assumption \P was not flush. As far as $v$ is concerned, if the turn of \P at $v$ onto the post-edge is not~\th, then \P has a similar (both \cw and \ccw) freedom at $v$; if the turn at $v$ is exactly~\th, then the freedom at $v$ is one-sided, i.e., $\cal{E}'$ has $0$ as an endpoint. Still, even in the latter case, $c$ can be rotated \cw around $v$ by a non-zero \eps while keeping the path feasible.

Let $\eps^*$ be an endpoint of $\cal{E}'$, such that the turn at $v$ onto the post-edge is not~\th and is to the right. (If the turn of \P at $v$ was~\th, then $\eps^*$ is the non-zero endpoint of $\cal{E}'$; otherwise, $\eps^*$ can be either of $\cal{E}'$'s endpoints.)

If $\eps^*$ is an endpoint of \cal{E}, then we can feasibly perturb \P so that $u,b$ and $c$ become collinear. If any of the arcs $X,Y$ has more than 1 edge, then collinearity of the arcs endpoints implies existence of an inflection edge in $P(\eps^*)$, which means that $P(\eps^*)$ can be shortened, and hence (since $P(\eps^*)$ and \P are equal-length) that \P was not a shortest path. On the other hand, if both $X=ub$ and $Y=bc$, then the collinearity of $u,b,c$ implies that $uc$ is a long edge of $P(\eps^*)$, which also means, by Lemma~\ref{lem:monotL}, that $P(\eps^*)$ is not shortest.

If $\eps^*$ is not an endpoint of \cal{E}, then an angular constraint must become tight at $\eps^*$  --- either the turn at $v$ becomes $\th$ and to the left (and then proceed as in the first paragraph), or the turn at one of $u(\eps^*), b(\eps^*),c(\eps^*)$ becomes~\th. In the latter case, either the path is flush at $u$, or the number of arcs in $P(\eps^*)$ becomes less than~3.

\section{Proof of Lemma \ref{lem:canoncial}}\label{app:lem:canonical}
Since we add turns of degree 0, the turn constraints and the turn-over-length constraints are not violated. To check the length constraints assume for the sake of contradiction that there exist adjacent short edges $ab,bc$ in~$\P'$.

If the turn at $b$ is not 0, at least one of the vertices $a,c$ must have been appended to \P (for otherwise already the original path \P had adjacent short edges $ab,bc$); suppose it is $a$ (Fig.~\ref{fig:canonical}, left). Let $a^-$ be the vertex of \P preceding $a$. By definition, no edge may contain two bridges; thus $ab$ is the only bridge on the edge $a^-b$, and the part $a^-a$ is a normal edge of some arc. Hence the edge $a^-b$ is long. By similar argument, if $c$ is appended, then $bc^+$ is long. If $c$ is not appended, then $bc$ is a short edge of the original path. In any case \P had a long edge $a^-b$ adjacent to a short or long edge $bc$ -- contradiction to Lemma~\ref{lem:ln}.

Assume now that the turn at $b$ is 0 (Fig.~\ref{fig:canonical}, right). Then $b$ has been appended, and thus is a bridge endpoint. An edge can contain at most two appended vertices, so one of $a,c$ is an original vertex of \P; suppose it is $c$. But then, since $b$ is an endpoint of an arc, $bc$ is a normal edge -- contradiction to our assumption that $ab,bc$ are short.
\begin{figure}\centering\includegraphics{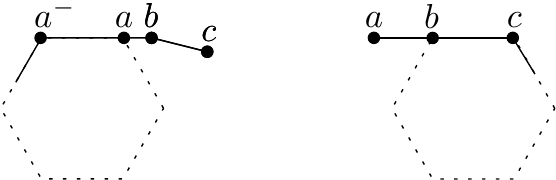}\caption{Left: If the turn at $b$ is not zero, then $a^-$ is a vertex of \P and $a^-a$ is a normal edge -- part of an arc, so \P had a long edge $a^-b$ adjacent to short edge $bc$. Right: If the turn at $b$ is zero, then $bc$ is a normal edge -- part of an arc.}\label{fig:canonical}\end{figure}

\section{Proof of Lemma~\ref{lem:bridge}}\label{app:lem:bridge}
Let us check that all constraints in the definition of a discrete curvature-constrained path are satisfied for $\P(\eps T)$ for small enough \eps.
\begin{list}{}{}
\item[\e{Turn constraints}]The bridge endpoints are the only vertices the turns at which may change due to the modification of the path. From definition of arc, a bridge never makes an angle \th with an adjacent edge (or else part of the bridge would have belonged to an arc). Thus, the turns at both $a$ and $b$ can change without violating the turn constraints, provided the changes are small.
\item[\e{Length constraints}]The bridge is the only edge whose length may change due to the modification; for that to possibly break the length constraints, the bridge must be adjacent to a short edge, and must itself become short in $P(\eps T)$ while being normal in \P (otherwise, if the bridge was long in \P, it will remain long also in $P(\eps T)$, for a sufficiently small \eps). However, having a normal bridge adjacent to a short edge contradicts optimality of \P. To see this, let us look more closely when a path can have a normal bridge:
    \begin{itemize}
    \item If both $a$ and $b$ were vertices of \P, then $ab$ would not have been a bridge (it would have been part of an arc).
    \item If none of $a,b$ was a vertex of \P (i.e., if both were added during the canonization), then $ab$ is adjacent to normal edges on both sides (the edges, together with the bridge, formed a length-3\ll edge of \P with turns of \th at both vertices); Fig.~\ref{fig:Bfreedom}, middle.
    \item If $a$ was originally a vertex of \P but $b$ was added during the canonization (Fig.~\ref{fig:Bfreedom}, right) then the edge $bb^+$ incident to $b$ is normal, so the only remaining problematic case is when the edge $a^-a$, incident to $a$ is short. But in this case the original path had a short edge adjacent to a long edge $bb^+$ contradictory to Lemma~\ref{lem:ln}.
    \item The situation when $b$ was a vertex of \P but $a$ was added during the canonization is symmetric to the above.
    \end{itemize}
\item[\e{Turn-over-length constraints}]Again, the only problematic case is when the bridge was normal in \P, became short in $P(\eps T)$, and the turn from $a^-a$ to $bb^+$ is larger then \th. Similarly to the above, if the bridge was normal, then its at least one incident turn was 0 in \P; hence for small \eps the turn-over-length constraints will not be violated.
\end{list}

\section{Proof of Lemma \ref{lem:b}}\label{app:lem:b}
\begin{list}{}{}
\item[\e{Proof of~\ref{item:bi}.}]Let $ab$ be a bridge, and $cd$ and inflection edge. First suppose that $ab$ and $cd$ are adjacent, say $b=c$ (Fig.~\ref{fig:bi}, left). If $b$ were not a vertex of \P before the canonization (i.e., if $b$ was added in the middle of an edge of \P), then $bd$ is not an inflection edge (the turn at $b$ is 0); so $b$ must be a vertex of \P. If $ab$ is not normal, then it cannot be adjacent to an inflection edge by Lemma~\ref{lem:infl}; so $ab$ must be normal. Then $a$ is not a vertex of \P, since otherwise $ab$, being a normal edge, would not have been a bridge (it would have been part of an arc). Thus, $a$ was added during the canonization in the middle of an edge $a^-b$ of \P. Since $|a^-a|=\ll$, we have that $ab$ is long, and is adjacent to an inflection edge $cd$ -- a contradiction to Lemma~\ref{lem:infl}.

Suppose now that $ab$ and $cd$ are not adjacent, and let $P'$ be the subpath of \P between $b$ and $c$ (Fig.~\ref{fig:bi}, right). We rigidly translate $P'$ so that $c$ slides towards $d$ and $ab$ rotates around $a$ keeping connectivity to $b$. By Lemmas~\ref{lem:IlengthFreedom} and~\ref{lem:bridge}, the path remains feasible. By the triangle inequality, the path shortens.
\begin{figure}
\begin{minipage}[c]{.5\columnwidth}\centering\includegraphics{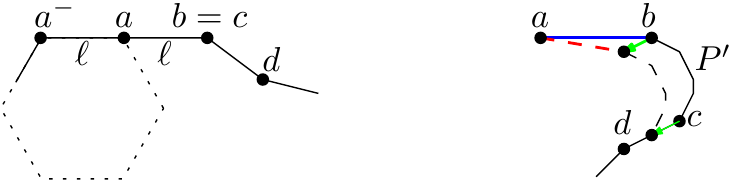}\caption{Left: $b=c$. It must be that $b$ is a vertex of \P and $ab$ is normal, implying that $a$ was added during canonization, and that \P had a long edge $a^-b$ adjacent to an inflection edge $bd$ -- contradiction to Lemma~\ref{lem:infl}. Right: The modification replaces {\color{blue}{blue}}+{\color{green}{green}} by {\color{red}{red}}.}\label{fig:bi}\end{minipage}\hfill
\begin{minipage}[c]{.4\columnwidth}\centering\includegraphics{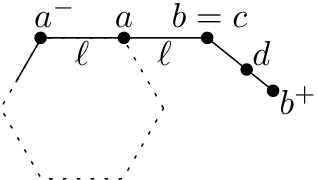}\caption{One of the adjacent bridges is normal, and hence is a proper subset of a long edge; the long edge is adjacent to a normal edge, and the second bridge has ``no space'' to fit on the normal edge.}\label{fig:2b}\end{minipage}
\end{figure}
\item[\e{Proof of~\ref{item:lb}}.]Similar to proof of~\ref{item:bi}: If the bridge and the long edge are adjacent, the bridge must be normal (by Lemma~\ref{lem:ln}), but then it belongs to a long edge of the original (non-canonized) path, which implies that the original path had adjacent long edges -- a contradiction to Lemma~\ref{lem:ln}. If the bridge and the long edge are not adjacent, the part between them can be moved along the long edge just as it was moved along the inflection edge in the proof of~\ref{item:bi}; Lemma~\ref{lem:Lfreedom} is used in place of Lemma~\ref{lem:IlengthFreedom} -- to ensure feasibility of the motion.
\item[\e{Proof of~\ref{item:2b}}.]Similar to proofs of~\ref{item:bi} and~\ref{item:lb}: First suppose that the two bridges $ab,cd$ are adjacent, i.e., $b=c$ (Fig.~\ref{fig:2b}). Since no edge can have two bridges, $b$ must be a vertex of the original path, which implies (Lemma~\ref{lem:ln}) that at least one of the bridges, say $ab$, is normal. The normal bridge must be a proper subset of a long edge $a^-b$ of the original path. Let $bb^+$ be the edge of the original path adjacent to $a^-b$; by Lemma~\ref{lem:ln}, $bb^+$ is normal. Now, if $d\ne b^+$, then $bd$ would be a short edge adjacent to a long edge $a^-b$ (contradicting Lemma~\ref{lem:ln}); on the other hand, if $d=b^+$ then $bd$ cannot be a bridge -- it should be (part of) and arc.

If the two bridges are not adjacent, the part between them can be moved along one of them just as it was moved along the inflection edge in the proof of~\ref{item:bi}; Lemma~\ref{lem:bridge} is used to ensure feasibility of the motion.
\end{list}\qed

\end{document}